\documentclass[%
onecolumn,
notitlepage,
superscriptaddress,
nofootinbib,
aps,
pra,
showkeys,
a4paper,
10pt%
]{revtex4-2}
\usepackage[utf8]{inputenc}
\usepackage{xcolor}
\usepackage{amsmath,amssymb,bm,cancel}
\usepackage{graphicx}
\graphicspath{{figures/}}

\usepackage{amsthm}
\usepackage[colorlinks,allcolors=black]{hyperref}
\usepackage[noabbrev,capitalise]{cleveref}

\newcommand{\R}{\mathbb{R}}
\newcommand{\lie}{\mathcal{L}}

\newcommand{\id}{\text{id}}
\newcommand{\Span}{\text{span}}

\newcommand{\inclusion}{i}
\newcommand{\boundary}{j}
\newcommand{\btob}{\partial i}
\newcommand{\chipboundary}{{\jmath}}
\newcommand{\intprod}{\iota}
\newcommand{\timelike}{\mathcal{T}}
\newcommand{\hamlike}{\mathcal{X}_{\mathcal{H}}}
\newcommand{\proj}{\text{proj}}

\newcommand{\Z}{\mathbb{Z}}
\newcommand{\C}{\mathbb{C}}

\DeclareMathOperator{\Diff}{Diff}

\DeclareMathOperator{\Int}{int}
\DeclareMathOperator{\supp}{supp}
\newcommand{\M}{\mathcal{M}}
\newcommand{\quotient}[2]{{\left.\raisebox{.2em}{$#1$}\middle/\raisebox{-.2em}{$#2$}\right.}}

\newcommand{\Y}{Y}

\theoremstyle{plain}
\newtheorem{theorem}{Theorem}[section]
\newtheorem*{theorem*}{Theorem}
\newtheorem{proposition}[theorem]{Proposition}
\newtheorem{corollary}[theorem]{Corollary}
\newtheorem*{conclusion*}{Main Conclusion}
\newtheorem{lemma}[theorem]{Lemma}

\theoremstyle{definition}
\newtheorem{definition}[theorem]{Definition}

\theoremstyle{remark}
\newtheorem{example}[theorem]{Example}
\newtheorem{remark}[theorem]{Remark}

\begin{document}

\title{Global realisation of magnetic fields as 1\texorpdfstring{$\frac{1}{2}$}{.5}D Hamiltonian systems
%on compact 3-manifolds with global Poincaré section
}

%% In alphabetical order
\author{N. Duignan}
\affiliation{School of Mathematics and Statistics, University of Sydney, NSW 2050,
Australia}

\author{D. Perrella}
\author{D. Pfefferlé}%
\email[]{david.pfefferle@uwa.edu.au}

\affiliation{The University of Western Australia, 35 Stirling Highway,
  Crawley WA 6009, Australia}

\begin{abstract}
The paper reviews the notion of $n+\frac{1}{2}$D non-autonomous Hamiltonian systems, portraying their dynamics as the flow of the Reeb field related to a closed two-form of maximal rank on a cosymplectic manifold, and naturally decomposing into time-like and Hamiltonian components. The paper then investigates the conditions under which the field-line dynamics of a (tangential) divergence-free vector field on a connected compact three-manifold (possibly with boundary) diffeomorphic to a trivial fibre bundle over the circle can be conversely identified as a non-autonomous $1\frac{1}{2}$D Hamiltonian system. Under the assumption that the field is transverse to a global compact Poincaré section, an adaptation of Moser's trick shows that all such fields are locally-Hamiltonian. A full identification is established upon further assuming that the Poincaré sections are planar, which crucially implies (together with Dirichlet boundary conditions) that the cohomology class of the generating closed one-forms on each section is constant. By reviewing the classification of fibre bundles over the circle using the monodromy representation, it is remarked that as soon as the Poincaré section of an alleged field is diffeomorphic to a disk or an annulus, the domain is necessarily diffeomorphic to a solid or hollow torus, and thus its field-line dynamics can always be identified as a non-autonomous Hamiltonian system. 
%The subtle point in the case of the hollow torus is that the Poincaré section may differ significantly from the purely poloidal or purely toroidal cross-sections, yet still be ``straightened" back to the canonical case after a suitable diffeomorphism. When the planar Poincaré section has more than one hole, for example in the case of the pair of pants, the underlying manifold is not necessarily a trivial fibre bundle. In such case, a direct identification with a non-autonomous Hamiltonian system is not possible.
\end{abstract}
  
\maketitle

% %%
% \section*{Journals}
% \begin{itemize}
%     \item \href{https://www.sciencedirect.com/journal/physica-d-nonlinear-phenomena}{Physica D: nonlinear phenomena}
%     \item \href{https://www.sciencedirect.com/journal/journal-of-geometry-and-physics}{Journal of Geometry and Physics}
%     \item \href{https://iopscience.iop.org/journal/1751-8121}{Journal of Physics A: Mathematical and Theoretical}
%     \item Journal of Mathematical Physics
%     \item Journal of Plasma Physics
%     \item \href{https://iopscience.iop.org/journal/0951-7715}{Nonlinearity}
% \end{itemize}

\section{Introduction}

Magnetic confinement fusion devices are shaped like tori because such domains are the ``simplest'' orientable manifolds in Euclidean three-space on which a vector field, tangent to the boundary, has a chance of being nowhere zero~\cite{thurston-1997}.
The layout of the magnetic field-lines in a fusion device determines the confinement~\cite{spitzer-1958,solovev-shafranov,kadomtsev-pogutse}.
The less chaotic their trajectories, the less heat and particles fusion devices will lose. This is the qualitative reason why configurations that exhibit nested toroidal flux-surfaces are sought after~\cite{cary-shasharina-1997,wakatani-1998,wobig-1999}.
Within magneto-hydrostatics (MHS), ensuring that the magnetic field-lines lie tangential to nested surfaces has been conjectured by Grad to be impossible without symmetry by isometry~\cite{grad-1967}. Certainly, configurations featuring three-dimensional flux-surfaces, as desired in stellarators, can only be explained through extreme fine-tuning and optimisation. Sensitivity/robustness to small changes is then an important area of investigation. Many equilibrium solvers and design/optimisation codes assume straight field-line coordinates and nested flux-surfaces as a starting point.

The magnetic field $B$ is considered to be a divergence-free vector field on a domain of Euclidean three-space, $M\subseteq \R^3$. One then defines the magnetic two-form through the natural volume-form $\mu$ as $\beta=i_B\mu$. The magnetic field $B$ being divergence-free is equivalent to the two-form $\beta$ being closed, $d\beta=0$. The latter statement is looser, in the sense that it does not depend on the metric. How best to express the closed two-form $\beta$ in order to highlight the features of the field-lines of $B$ is the focus of this paper.

Estimating the degree of chaos in magnetic field-lines is hard other than through numerical field-line tracing. There are plenty of examples of nowhere vanishing vector fields on three-manifolds with extremely rich field-line dynamics. The latter dynamics can be as complicated as desired, for example with ABC-flows~\cite{arnold-1965}, Turing complete flows~\cite{cardona-2021}, knotted flows~\cite{encisco-peralta-salas-2015}, etc.
Analytically, the tools of Hamiltonian perturbation theory are often employed. For example, the Chirikov criterion predicts the onset of chaos when individual resonances overlap~\cite{chirikov-1960}. The KAM theorem~\cite{kolmogorov-1954,arnold-1963a,moser-1967} is invoked to assess the persistence of irrational surfaces under small perturbations~\cite{greene-1979,meiss-1992,mackay-percival,hudson-breslau,dewar-2012}.
Magnetic field-lines are thus often compared to non-autonomous $1\frac{1}{2}$D Hamiltonian systems~\cite{kerst-1962,cary-littlejohn-1983,cary-hanson,boozer-1983}, as highlighted in the following direct quote from \cite[Section 4.5 a)]{dhaeseleer}.
\begin{quote}\it
The equation of a field line can be cast in a Hamiltonian form. Hamiltonian theory guarantees the existence of magnetic surfaces in devices with a symmetry direction: an ideal, axisymmetric tokamak; a helically symmetric straight stellarator; etc. In non-symmetric devices, on the other hand, magnetic surfaces do not exist rigorously everywhere because ergodic regions of finite volume exist [\ldots]. In practice, we can get away with considering approximate magnetic surfaces provided the perturbations away from symmetry are not too large. In principle, one should either measure the surfaces experimentally or follow field lines on a computer to demonstrate the existence of suitable surfaces.
\end{quote}
The justification in the physics literature is local, gauge dependent and requires explicit coordinates~\cite{helander-2014,meiss-1992}.
A variant of Darboux's theorem can be used in Euclidean three-space to show that a non-zero magnetic field can be represented in \emph{Clebsch form}, $B=\nabla f \times\nabla g$, in a neighbourhood of any point. The result is local in nature, and there are topological obstructions to realising Clebsch representations globally.

Under the assumption that the magnetic field is transverse to the level sets of a circle-valued \emph{time} or \emph{angle} function, we show using Moser's trick~\cite{moser-1965} under what circumstances a magnetic field can be identified with a $1\frac{1}{2}$D Hamiltonian system. The procedure is coordinate-free and global. This identification holds on compact three-manifolds diffeomorphic to a product $S^1\times \Sigma$. This includes the domains of interest to magnetic confinement such as the solid torus and the hollow torus. There is an opportunity to extend the results to non-trivial bundles, for example generated by excluding linked loops from a toroidal domain as in heliotrons~\cite{wakatani-1998} or in the CNT device~\cite{pedersen-boozer-2002}, or by removing loops in integer multiples of the generating class of first homology.

The paper is organised as follows. In \cref{sec:beltrami-example}, an example of a \emph{Beltrami} magnetic field on a hollow torus is shown to possess a Poincaré section that cannot purely be the level sets of the toroidal angle or poloidal angle, but a combination of both. In \cref{sec:local-coordinates,sec:weyl-gauge}, direct methods to express the magnetic field in the so-called double-Clebsch form are explored. The first method is local and coordinate-dependent, and the second is global and coordinate-free relying on the so-called Weyl gauge. The drawback of these direct methods is that they only apply on specific domains. In \cref{sec:hamiltonian-systems}, we review the notion of non-autonomous Hamiltonian systems in the canonical setting and reveal their relation to cosymplectic manifolds. Necessary conditions for dynamics on cosymplectic manifolds to admit a Hamiltonian representation are determined. In \cref{sec:divergence-free-fields}, we specialise to the case of divergence-free fields on domains of $\R^3$ with tangential boundary conditions and present a procedure involving Moser's trick to decompose the magnetic two-form into a symplectic form on the fibres and the wedge product of a Hamiltonian function and the derivative of the circle-valued projection. \cref{sec:conclusion} discusses the results and hints to possible extensions.

\section{Motivating example and direct methods}
The main strategy to relate field-line equations to Hamiltonian systems is to arrive via coordinate transformations at the commonly adopted \emph{double-Clebsch} representation of the magnetic field on a toroidal domain of $\R^3$,
\begin{align*}
    B = \nabla\psi_t\times\nabla \theta + \nabla\varphi \times \nabla \psi_p(\psi_t,\theta,\varphi),
\end{align*}
where $(\psi_t,\theta,\varphi)$ are known as \emph{toroidal flux coordinates} in the literature~\cite{dhaeseleer,helander-2014}; $\psi_t$ is a radial variable (measuring the toroidal flux when the field-lines are nested), $\theta$ is a poloidal angle wrapping around the short way of the torus and $\varphi$ is a toroidal angle travelling around the long way. The function $\psi_p$ yields the poloidal flux when field-lines are nested. In this picture, $\psi_t$ is seen as the canonical action paired to the angle $\theta$ and $\psi_p$ is identified as a ``time-dependent" Hamiltonian function.

There are two distinct questions when proceeding in this way; how are these coordinates (globally) recovered for an arbitrary magnetic field, and do they necessarily lead to a Hamiltonian representation of the field-line dynamics ?

It is useful to note that the toroidal component of the magnetic field,
\begin{align}\label{eq:double-clebsch}
B^\varphi=B\cdot\nabla\varphi = (\nabla\psi_t\times\nabla\theta)\cdot\nabla\varphi = \frac{1}{\sqrt{g}},
\end{align}
is inversely proportional to the Jacobian determinant of the coordinate transformation. By the inverse function theorem, the transformation is (locally) invertible if and only if $B^\varphi\neq 0$ is nowhere zero, namely the magnetic field admits (global) Poincar\'e sections given by the level sets of $\varphi$, also known as poloidal cross-sections. In the literature of magnetic confinement fusion, the assumptions around Poincar\'e sections are often taken for granted. There are examples of devices where the toroidal component changes sign, for example in the Reversed-Field Pinch (RFP) \cite{lorenzini-2009}. The magnetic fields in such devices cannot admit a global Poincaré section, preventing the identification of a forward ``time" direction and therefore the strict association of the field-line dynamics with a Hamiltonian system. While it is possible to glue coordinate patches in which magnetic fields are expressed as above and a consistent forward time-direction is found, the procedure is bound to fail under global perturbations.

Below is an example of a nowhere-vanishing force-free (Beltrami) magnetic field whose toroidal and poloidal components vanish at different locations, which seems to suggest that the identification with a Hamiltonian system might (globally) fail. After an elementary redefinition of the angles via a global coordinate transformation, it is shown that the field actually admits a global Poincaré section, and is therefore genuinely Hamiltonian, albeit with respect to non-standard poloidal and toroidal angles. This example highlights part of the challenges pertaining to Poincaré sections. The rest of the paper will be concerned about deeper obstructions.

\subsection{A Lundquist-Beltrami field with ``full" rotational transform}
\label{sec:beltrami-example}
Eigenfields of the curl operator, depicting Beltrami fields or linear force-free fields, were actively studied in astrophysical contexts \cite{lundquist1951stability, chandrasekhar1957force} (see also \cite{Yoshida_Giga_1990, Yoshida_1992}). Such vector fields are relevant in plasma physics and magnetic confinement fusion through a process called \emph{Taylor relaxation}~\cite{Taylor_1986}, whereby the magnetic energy is dissipated on faster timescales than magnetic helicity. As such, Beltrami states are minimisers of a quadratic energy functional, and thus solutions of an elliptic variational problem. They are also special cases of magneto-hydrostatics (MHS) solutions (also called MHD equilibrium). Examples of Beltrami fields can easily be constructed in simple slab or cylindrical geometries. The question is whether the field-line dynamics can be cast as a non-autonomous Hamiltonian system. As with the elementary ABC-flows~\cite{arnold-1965}, the identification is not generally expected to hold. An important obstruction is highlighted, but then resolved, in the following example. 

Let $J_0,J_1 : \R \to \R$ denote the Bessel functions of order $0$ and order $1$ respectively. For $0 < r_0 < r_1<\infty$, consider the annulus
\begin{equation}\label{eq:annulus}
A_{r_0,r_1} = \{(x,y) \in \R^2 :  r_0 \leq \sqrt{x^2+y^2} \leq r_1\}. 
\end{equation}
Then, on the \emph{periodic co-axial cylinder} $M = S^1 \times A_{r_0,r_1}$, where $S^1 = \quotient{\R}{2\pi\Z}$, the vector field
\begin{equation}
\label{eq:beltrami-example}
B = \frac{J_1(r)}{r}\frac{\partial}{\partial \theta}\Big|_{r,\varphi} - J_0(r)\frac{\partial}{\partial\varphi}\Big|_{r,\theta}
= -\nabla (r J_1(r))\times \nabla\theta + \nabla J_0(r)\times \nabla\varphi ,     
\end{equation}
satisfies, with respect to the metric $g=d\varphi^2+ dr^2+r^2d\theta^2$,
\begin{equation*}
\nabla \times B = -B, \qquad B \cdot n = 0 \quad \text{(on $\partial M$)},    
\end{equation*}
where $n$ denotes the outward unit normal on $M$ and $(\varphi,r,\theta)$ are coordinates coming from the projection $\varphi : M \to S^1$ and usual polar coordinates. The vector field $B$ is the same as that considered by \citet{lundquist1951stability} but restricted to a different domain.

Let $z_{i,1}$ denote the first positive zero of $J_i$, $i \in \{0,1\}$. As well-known, $0 < z_{0,1} < z_{1,1}$, and so we may set
\begin{equation*}
r_0 = z_{0,1}, \qquad r_1 = z_{1,1}.
\end{equation*}
The  field-lines of $B$ satisfy the following ODEs
\begin{align*}
    \dot{r}(t) &= 0,\\
    \dot{\theta}(t) & = \frac{J_1(r(t))}{r(t)},\\
    \dot{\varphi}(t) &= - J_0(r(t)),
\end{align*}
yielding excessively simple dynamics. Any function of the coordinate $r$ forms an invariant or \emph{flux-function}. The pitch of the field-lines in these coordinates is constant, and equal to the \emph{winding number}, also known as \emph{rotational transform}~\cite{dhaeseleer},
\begin{align*}
    \iota(r) =\frac{\dot\theta}{\dot\varphi}= \frac{B\cdot\nabla\theta}{B\cdot\nabla\varphi} = -\frac{J_1(r)}{rJ_0(r)}.
\end{align*}
We would like to see if the field-lines dynamics can be described as a non-autonomous $1\frac{1}{2}$D Hamiltonian system. The necessary condition that $B$ is nowhere zero is met, as shown later in \cref{eq:positive-along-eta}. An obstacle is encountered in that neither the $\varphi$ coordinate nor the $\theta$ coordinate serve as an angle against which the field $B$ is positive. In other words, the magnetic field has a rotational transform ranging from $0$ to $\infty$. As well-known, writing $z_{0,2}$ for the second positive zero of $J_0$, we have $z_{0,1} < z_{1,1} < z_{0,2}$ and $J_0|_{(z_{0,1},z_{0,2})} < 0$. On the other hand, $J_1|_{(0,z_{1,1})} > 0$. So, $J_0(z_{0,1}) = 0 = J_0(z_{0,2})$ and $J_1(z_{1,1}) = 0$. Therefore,
\begin{align*}
\{x \in M : B\cdot\nabla\varphi|_x > 0\} &= M \setminus \partial_0 M,\\
\{x \in M : B\cdot\nabla\theta|_x > 0\} &= M \setminus \partial_1 M,
\end{align*}
where $\partial_0 M$ is the inner connected component of $\partial M$,
\begin{equation*}
\partial_0 M = S^1 \times \{(x,y) \in \R^2 : x^2+y^2 = r_0^2\},
\end{equation*}
and $\partial_1 M$ is the outer connected component of $\partial M$,
\begin{equation*}
\partial_1 M = S^1 \times \{(x,y) \in \R^2 : x^2+y^2 = r_1^2\}.
\end{equation*}
That is, neither $B\cdot\nabla\varphi > 0$ nor $B\cdot\nabla\theta > 0$ on the whole of $M$. In fact, even if one introduces smooth functions $f,g$ on $M$, and considers angle functions $\theta', \varphi' : M \to S^1$ such that
\begin{align*}
\nabla\theta' &= \nabla(\theta + f)=\nabla\theta+\nabla f,&
\nabla\varphi'&= \nabla(\varphi + g)=\nabla\varphi+\nabla g,
\end{align*}
then one finds that neither $B\cdot\nabla\theta' > 0$ nor $B\cdot\nabla \varphi' > 0$. Indeed, one has
\begin{equation*}
B\cdot\nabla\theta'|_{\partial_1 M} = B\cdot\nabla\theta|_{\partial_1 M} + B\cdot\nabla f|_{\partial_1 M} = B\cdot\nabla f|_{\partial_1 M},
\end{equation*}
and
\begin{equation*}
B\cdot\nabla\varphi'|_{\partial_0 M} = B\cdot\nabla\varphi|_{\partial_0 M} + B\cdot\nabla g|_{\partial_0 M} = B\cdot\nabla g|_{\partial_0 M}.
\end{equation*}
But, for any function $h \in C^{\infty}(M)$ and any $i \in \{0,1\}$, because $h|_{\partial_i M}$ attains a maximum on $\partial_i M$ and therefore has a critical point on $\partial_i M$, the fact that $B$ is tangent to $\partial_i M$ implies that $B\cdot\nabla h|_{\partial_i M}$ vanishes at some point in $\partial_i M$. Therefore, modifying the angle functions $\theta,\varphi$ by smooth functions $f,g$ is futile.
\begin{figure}
    \centering
    \includegraphics[width=0.6\linewidth]{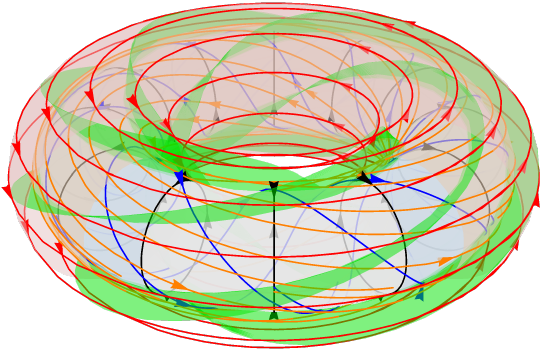}
    \caption{Field-lines and flux-surfaces of the Beltrami field of \cref{eq:beltrami-example}, represented in toroidal geometry (a non-isometric embedding). On the (grey) inner boundary, the (black) field-lines are purely poloidal and on the (red) outer boundary, the (red) field-lines are purely toroidal. On an intermediate flux-surface, (blue/orange) field-lines feature irrational winding. The green surfaces are level-sets of the $\phi=\theta+\varphi$ circle-valued function. The field-lines are everywhere transverse to these surfaces, so that they constitute global Poincaré sections, which is the necessary and sufficient requirement to identify the field-lines $B$ as a non-autonomous Hamiltonian system.}
    \label{fig:beltrami-example}
\end{figure}

This predicament is remedied by considering a new angle\footnote{Any integer combination of $\theta$ and $\varphi$ produces a circle-valued function by the group structure of $S^1$, whereas general $\R$-linear combinations do not yield a function $M \to S^1$.} $\phi = \theta + \varphi : M \to S^1$. To see this, first note that
\begin{equation*}
B\cdot\nabla\phi = \frac{J_1(r)}{r} - J_0(r).  
\end{equation*}
From the fact that $J_0|_{(z_{0,1},z_{0,2})} < 0$ and $J_1|_{(0,z_{1,0})} > 0$, it follows that for $r \in [z_{0,1},z_{1,1}]$, one has that $J_1(r)/r - J_0(r) > 0$. Thus,
\begin{equation}\label{eq:positive-along-eta}
B\cdot\nabla\phi > 0.
\end{equation}
In fact, let us consider the coordinate transformation $F : M \to M$ given by
\begin{align*}
(\varphi,r,\theta)\mapsto (\theta+\varphi,r,\theta).
\end{align*}
Then, the contravariant components of the magnetic field in these new coordinates are computed as
\begin{align*}
B^\phi=B\cdot\nabla\phi &= \frac{J_1(r)}{r} - J_0(r)>0,\\
B^r=B\cdot\nabla r &= 0,\\
B^\theta =  B\cdot\nabla\theta &= \frac{J_1(r)}{r},
\end{align*}
so that
\begin{equation*}
B = \frac{J_1(r)}{r}\frac{\partial}{\partial\theta}\Big|_{r,\phi} + \left(\frac{J_1(r)}{r} - J_0(r)\right)\frac{\partial}{\partial\phi}\Big|_{r,\theta}
= -\nabla(J_0+rJ_1)\times \nabla \theta + \nabla J_0\times \nabla\phi,
\end{equation*}
One can then identify the flow of $B$ as the dynamics of a non-autonomous Hamiltonian system by following the approach in \cref{sec:hamiltonian-systems}.

In the remainder of the paper, we consider in higher generality what representing a magnetic field as a non-autonomous Hamiltonian system entails and show that diffeomorphisms such as the one above can be found in a wide variety of cases. A direct consequence of our investigation applied to the case of a hollow torus or solid torus is the following list of equivalent statements, presented here as amuse-bouches. To set up, let $D^2$ denote the (closed) unit disk in $\mathbb{R}^2$.
\begin{conclusion*}
\label{res:advert}
Let $M$ be a compact connected manifold with boundary, $\mu$ a volume-form on $M$, and $B$ a divergence-free vector field on $M$ tangent to $\partial M$. Let $\Sigma \in \{D^2,A_{1,2}\}$. Then, the following are equivalent. 
\begin{enumerate}
    \item\label{point:advert-poincare} $B$ has a Poincar\'e section diffeomorphic to $\Sigma$;
    \item\label{point:advert-oneform} $M$ is diffeomorphic to $S^1 \times \Sigma$ and there exists a closed one-form $\eta \in \Omega^1(M)$ such that $\eta(B) > 0$;
    \item\label{point:advert-hamiltonian} The field-lines of $B$ can be represented as a non-autonomous $1\frac{1}{2}$D Hamiltonian system on $\Sigma$.
\end{enumerate}
\end{conclusion*}
We clarify and recall the meaning of \cref{point:advert-hamiltonian} in \cref{sec:hamiltonian-systems} and \cref{point:advert-poincare} in \cref{sec:divergence-free-fields}.

\subsection{Local flux-coordinates and Hamiltonian representation}
\label{sec:local-coordinates}
There are several local procedures to derive flux-coordinates as in \cref{eq:double-clebsch}. Following \cite{whiteman-1977}, let $(u^1,u^2,u^3)$ be local coordinates of a domain of $\R^3$ with volume form $\mu=\sqrt{g}du^1\wedge du^2\wedge du^3$. First, the magnetic field is expressed in a \emph{covariant} representation as
\begin{align*}
    B = B^i \frac{\partial}{\partial u^i} \implies \beta =i_B\mu
    = \sqrt{g} B^1 du^2\wedge du^3 + \sqrt{g} B^2 du^3\wedge du^1 +\sqrt{g} B^3 du^1\wedge du^2
\end{align*}
with $B^i=B\cdot\nabla u^i$. The condition that the magnetic field is divergence-free (equivalently that the two-form $\beta$ is closed) reads
\begin{align*}
\partial_i (\sqrt{g} B^i) = 0.
\end{align*}
We wish to establish a coordinate transformation $(u^1,u^2,u^3)\mapsto (q,p,t)$ defined by
\begin{align*}
    q &= Q(u^1,u^2,u^3)=u^1, &
    p &= P(u^1,u^2,u^3), &
    t &= T(u^1,u^2,u^3)=u^3,
\end{align*}
such that the magnetic two-form has exactly the following expression
\begin{align*}
    \beta = dP\wedge dQ + dT\wedge dH 
    = - \partial_2 H du^2\wedge du^3 
    +(\partial_3 P +\partial_1 H) du^3\wedge du^1
    -\partial_2 P du^1\wedge du^2
\end{align*}
for some functions $H$ and $P$ of $(u^1,u^2,u^3)$. 
%With volumne form $\mu=J dq\wedge dp\wedge dt$ and $B=B^q\partial_q + B^p\partial_p+B^t\partial_t$, we would have
% \begin{align*}
%     \beta = -J B^q dt\wedge dp + J B^p dt\wedge dq + J B^t dq\wedge dp \iff 
%     \left|
%     \begin{array}{rl}
%     B^t &= -J^{-1}, \\
%     B^q &= -J^{-1}\frac{\partial H}{\partial p} = B^t \frac{\partial H}{\partial p},\\
%     B^q &= J^{-1} \frac{\partial H}{\partial q} = -B^t \frac{\partial H}{\partial q},\\
%     \end{array}
%     \right.
% \end{align*}

By identification, the following relations between the functions $P$, $H$ and the covariant components of the magnetic field are obtained
\begin{align*}
    \partial_2 H &= - \sqrt{g} B^1, &
    \partial_2 P &= - \sqrt{g} B^3, &
    \partial_3 P + \partial_1 H &= \sqrt{g} B^2.
\end{align*}
If the first and second conditions are satisfied, then the derivative of the third with respect to $u^2$ is automatically verified. Defining
\begin{align*}
P(u^1,u^2,u^3) &= -\int_{a}^{u^2}\sqrt{g} B^3(u^1,v,u^3) dv + c_1(u^1,u^3) \\
    H(u^1,u^2,u^3) &= -\int_{a}^{u^2}\sqrt{g} B^1(u^1,v,u^3) dv + c_2(u^1,u^3),
\end{align*}
then, writing $\sqrt{g}B^2(u^1,u^2,u^3)=\sqrt{g}B^2(u^1,a,u^3)+\int_a^{u^2}\partial_2(\sqrt{g}B^2)(u^1,v,u^3)dv$, the double-Clebsch representation of the magnetic field is achieved by asserting that
\begin{align*}
   \partial_3 c_1(u^1,u^3) + \partial_1 c_2(u^1,u^3) = \sqrt{g}B^2(u^1,a,u^3).
\end{align*}
This is possible, for example, by letting $c_2=0$ and by setting
\begin{align*}
    c_1(u^1,u^3) = \int_b^{u^3}\sqrt{g}B^2(u^1,a,v)dv.
\end{align*}
The coordinate transformation is valid provided that the Jacobian matrix is invertible, or equivalently that the top-form
\begin{align*}
    dP\wedge dQ\wedge dT = \beta\wedge du^3
\end{align*}
is a volume-form (non-degenerate). We must thus have $\partial_2 P=-\sqrt{g}B^3\neq 0$, which assumes that the magnetic field is transverse to the level sets of $u^3$ acting as Poincaré sections. Often omitted and quite tedious in practice, the final step to relate the magnetic field to a Hamiltonian system is to express the function $H$ in terms of the new coordinates, namely to define the Hamiltonian $\hat{H}$ so that $\hat{H}(q,p,t) = H(u^1,u^2,u^3)$ when $(q,p,t)=(u^1,P(u^1,u^2,u^3),t)$. Then, one reaches the desired representation in the coordinates $(q,p,t)$
\begin{align*}
    \hat{\beta} = dp\wedge dq + dt\wedge d\bar{H}(q,p,t).
\end{align*}

\subsection{Global coordinate-free procedure exploiting the Weyl gauge}
\label{sec:weyl-gauge}
A coordinate-free and global version of the construction described in the previous section is when a smooth (strong) deformation retract $F:M\times[0,e)\to M$ with $e>0$ of the manifold $M$ onto a sub-manifold $L$ can be manufactured as the semi-flow of the (autonomous) vector field $X$ fitting the description below. The maps $F_s:M\to M$, $x\mapsto F(x,s)$ is a family of diffeomorphisms satisfying the ordinary differential equation
\begin{align*}
\partial_s F_s = X\circ F_s,
\end{align*}
with initial condition $F_0=\id$. It is clear that the semi-flow of $X$ enjoys the group property $F_a\circ F_b=F_{a+b}=F_b\circ F_a$ whenever $0\leq a+b<e$. This property implies in particular that the generating vector field is invariant under its semi-flow, ${F_s}_*X=X$ for all $s\in[0,e)$.
We now compute the $s$-derivative of the pullback by $F_s$ of the magnetic two-form $\beta=i_B\mu$,
\begin{align*}
    \frac{d}{ds} F_s^* \beta = F_s^*\lie_X\beta = d (F_s^* i_X\beta).
\end{align*}
Integrating with respect to $s$ and using the fundamental theorem of calculus, we have
\begin{align}
\label{eq:weyl-gauge}
    \beta  &= d\alpha + F^*_{e}\beta, & \text{where}&&\alpha &= -\int_0^{e} F_s^*i_X \beta ds.
\end{align}
If $L$ is a one-dimensional submanifold (e.g. $M$ retracts to a line, as is the case for the solid torus), then $F_{e}^*\beta=0\in\Omega^2(L)$. If $L\subset \partial M$ (e.g. $M$ retracts to a component of the boundary, as is the case for the hollow torus) and the magnetic field is assumed to be tangential to that boundary component, then $F_{e}^*\beta=0\in\Omega^2(L)$ because the two-form is Dirichlet, see \cref{lem:boundaryconditions}. In either case, a unique vector potential $\alpha$ is obtained to represent the (exact) magnetic two-form, $\beta=d\alpha$. The direct evaluation of the vector potential of \cref{eq:weyl-gauge} corresponds to a choice of gauge; it is an instance of the \emph{Weyl gauge}, also known as multi-polar gauge or Poincaré gauge. An interesting and defining property of the potential one-form $\alpha$ is that its component along $X$ is zero. Indeed,
\begin{align*}
    \alpha(X) = -\int_0^{e} (F_s^*i_X\beta)(X) ds 
    = -\int_0^{e} F_s^*(\beta(X,{F_s}_*X))ds 
    = -\int_0^{e} F_s^*(\beta(X,X))ds
    = 0.
\end{align*}

Now, to arrive at the double-Clebsch representation of $\beta$, we further assume that there exists a function $\rho$ and two linearly independent closed one-form $\kappa_1$ and $\kappa_2$ on $M\setminus L$, such that
\begin{align*}
    i_X\mu = -\rho \kappa_1\wedge \kappa_2.
\end{align*}
There is freedom in the choices of $\rho$, $\kappa_1$ and $\kappa_2$ beyond the fact that they must produce a well-defined deformation retract $F$ of $M$ onto $L$, which can be difficult to assess in practice. The regularity of $\kappa_1$ and $\kappa_2$ as closed one-forms on $M\setminus L$ may not extend to the whole manifold $M$, in which case the function $\rho$ has to be suitably tailored to compensate singularities in order for $X$ to be well-defined, see \cref{ex:weyl-solid-torus}.

Since $X\in \ker\kappa_1\cap \ker\kappa_2 $, we have $\lie_X\kappa_j=i_X\cancel{d\kappa_j}+d\cancel{i_X\kappa_j}=0$, which means that the semi-flow preserves the closed one-forms, namely $F_s^*\kappa_j=\kappa_j$ for all $s\in[0,e)$ and $j\in\{1,2\}$.

With the above prescription for $X$, an explicit expression for the potential one-form $\alpha$ is obtained. First, note that 
\begin{align*}
    -i_X\beta = -i_Xi_B\mu=i_Bi_X\mu 
    =-i_B( \rho\kappa_1\wedge \kappa_2)
    = -\rho\kappa_1(B)\kappa_2 + \rho \kappa_2(B)\kappa_1
    = \kappa_2(\rho B)\kappa_1 - \kappa_1(\rho B)\kappa_2.
\end{align*}
Because the flow preserves the one-forms $\kappa_1$ and $\kappa_2$, the potential one-form is then given by 
\begin{equation*}
\alpha = P\kappa_1 - H \kappa_2
\end{equation*}
where 
\begin{align*}
    P &= \int_0^e F_s^*(\kappa_2(\rho B))ds,&
    H &= \int_0^e F_s^*(\kappa_1(\rho B)) ds.
\end{align*}
Consequently, the two-form will be globally represented in the double-Clebsch form
\begin{equation*}
\beta = dP\wedge \kappa_1 - dH\wedge \kappa_2.
\end{equation*}
To recover a genuine Hamiltonian representation from here, one needs to identify canonical coordinates and a Hamiltonian function. For the same reasons as in \cref{sec:local-coordinates}, this (morally) works if the top-form
\begin{align*}
    dP\wedge \kappa_1\wedge\kappa_2 = \beta\wedge \kappa_2 
\end{align*}
is a volume-form. Then, $\kappa_2$ provides the \emph{time}-coordinate, $dP\wedge\kappa_1$ defines a symplectic form on the leaves (Poincaré sections) of the foliation integrated from $\ker \kappa_2$, and $H$ qualifies as a non-autonomous Hamiltonian function. The frailty of the method is highlighted in the following example.

\begin{example}
\label{ex:weyl-solid-torus}
    On a manifold diffeomorphic to a solid torus $M\cong S^1\times D^2$ with volume-form $\mu=\sqrt{g}dr\wedge d\theta\wedge d\phi$, consider the vector field $X=-r\partial_r$ which satisfies
    \begin{align*}
        i_X\mu =-r\sqrt{g} d\theta\wedge d\phi
    \end{align*}
    where $\phi:M\to S^1$ and $(r,\theta)$ are polar coordinates on $D^2$. Here, $\rho=r\sqrt{g}$, $\kappa_2=d\phi$ and $\kappa_1=d\theta$ which is a well-defined one-form on $M\setminus L$, namely away from the coordinate axis $L\cong S^1\times \{0\}$. The semi-flow of $X$ is given in coordinates by
    \begin{align*}
        F_s(r,\theta,\phi) = (r e^{-s},\theta,\phi)
    \end{align*}
    which adequately deformation retracts the solid torus to the coordinate axis $L$ in the limit $s\to\infty$. Then, the functions $P=\alpha(\partial_\theta)$ and $H=-\alpha(\partial_\phi)$ are the covariant components of the vector potential, uniquely determined by the contravariant components of the magnetic field  $B^\phi=d\phi(B)$ and $B^\theta=d\theta(B)$ through the following integrals
\begin{align*}
    P(r,\theta,\phi)  &= r \int_0^\infty e^{-s} \sqrt{g}(re^{-s},\theta,\phi) B^\phi(re^{-s},\theta,\phi) ds 
    = r  \int_0^1 \sqrt{g}(u r,\theta,\phi) B^\phi(u r,\theta,\phi) du \\
    H(r,\theta,\phi) &= r \int_0^\infty e^{-s} \sqrt{g}(re^{-s},\theta,\phi) B^\theta(re^{-s},\theta,\phi) ds 
    = r \int_0^1 \sqrt{g}(u r,\theta,\phi) B^\theta(u r,\theta,\phi) du.
\end{align*}
\end{example}
This method extends to the hollow torus by using a similar radial vector field $X\propto \partial_r$ under the additional assumption that the magnetic field is tangential to the boundary. The representation of the vector potential without a radial component is extremely common in plasma physics, in particular in numerical codes such as VMEC~\cite{hirshman-1983} and SPEC~\cite{hudson-2012}.

This method is limited to manifolds admitting i) a deformation retract onto the ``right" submanifold, ii) closed but nowhere-vanishing one-forms. For example, it does not seem possible to apply the Weyl gauge to the case of a solid torus in which two toroidal loops are removed (yielding a cross-section diffeomorphic to a disk with two holes). It is also not immediate that the functions $P$ and $H$ defined above on $M\setminus L$ are smooth on $M$. While $B^\theta=d\theta(B)$ is often singular at the coordinate axis, the product $\sqrt{g}B^\theta$ is usually integrable. For example, consider the magnetic field $B=\partial_x$ on a periodic cylinder with the flat metric $g=dx^2+dy^2+d\phi^2=dr^2+r^2d\theta+d\phi^2$ where $(x,y)$ are Cartesian coordinates of the disk, one finds that $B^\theta=-\sin\theta/r$, but $\sqrt{g}B^\theta=-\sin\theta$ so that $H=-r\sin\theta =-y$ is smooth. Overall, the method is still sensitive to the choice of coordinates and magnetic fields.

In the remainder of the paper, we explore a vastly different and possibly more robust method that circumvents some of the caveats discussed above. We first start by reviewing in precise terms what non-autonomous Hamiltonian systems means.

\section{Non-autonomous Hamiltonian systems}
\label{sec:hamiltonian-systems}

\subsection{Classical non-autonomous Hamiltonian systems in canonical coordinates}
\label{sec:classical-hamiltonian-systems}
  Classically, a non-autonomous Hamiltonian system is considered to be an ``$n+\tfrac12$" degree of freedom dynamical system on the cotangent bundle $T^*Q$ of a smooth $n$-manifold $Q$. It is defined by a \emph{time-dependent} Hamiltonian function $H: \R\times T^*Q\to \R$ or $H:S^1\times T^*Q\to \R$ for periodicity in time. We will work with the periodic case.
  
  In canonical coordinates $(q^i,p_i)$ on $T^*Q$, Hamilton's equations read
  \begin{align*}
      \dot{q}^i(t) &= \ \ \frac{\partial H}{\partial p_i}(t,q(t),p(t))\\
      \dot{p}_i(t) &=-\frac{\partial H}{\partial q^i}(t,q(t),p(t))
  \end{align*}
  which can be recorded as the dynamics of the non-autonomous vector field on $T^*Q$,
  \begin{align*}
      X_t = \frac{\partial H_t}{\partial p_i}\partial_{q^i} - \frac{\partial H_t}{\partial q^i}\partial_{p_i},
  \end{align*}
     where $H_t:T^*Q\to \R$ is now a family of Hamiltonian functions on $T^*Q$ defined by $H_t(q,p)=H(t,q,p)$ for $t\in S^1$.
  
  It is natural to view \emph{time} as an independent coordinate and elevate the non-autonomous system to an autonomous one on $S^1\times T^*Q$, given by the differential equations
  \begin{align*}
    \dot{t}(s) &= 1\\
    \dot{q}^i(s) &= \ \ \frac{\partial H}{\partial p_i}(t(s),q(s),p(s))\\
    \dot{p}_i(s) &=-\frac{\partial H}{\partial q^i}(t(s),q(s),p(s)).
  \end{align*}
  In turn, this can be thought of as the dynamics of the \emph{evolution vector field} $E_H$ on $S^1\times T^*Q$, which can be expressed as the sum
  \begin{align}
  \label{eq:canonical_evolution-field}
      E_H = \partial_t + X_H,
  \end{align}
  where the vector field $X_H$ on $S^1\times T^*Q$ evaluates to ${X_H}(t,q,p) = X_t(q,p)$ for  $t\in S^1$ and $(q,p)\in T^*Q$.

\subsection{Some trivial bundle formalism}
\label{sec:trivial-bundle-formalism}
Before diving into a modern representation of non-autonomous Hamiltonian dynamics, this section introduces some basic formalism around product manifolds and trivial bundles. A review of the classification of fibre bundles over the circle can be found in \cref{sec:MappingToriMCG}.

% Non-autonomous Hamiltonian systems naturally live on trivial bundles over the real line or over the circle, $M\cong S^1\times \Sigma$, where $\Sigma$ is here a $2n$-dimensional manifold possibly with boundary. This section revisits a few elementary properties of product manifolds to streamline upcoming proofs. 

A fibre bundle over the circle would be defined in terms of a projection $\phi : M \to S^1$. A trivial bundle $M=S^1\times \Sigma$ comes with natural projections onto each factor,
\begin{align}
  \phi&=\proj_1:S^1\times \Sigma\to S^1, \quad (t,x)\mapsto t\label{eq:projection_circle}\\
  \pi&=\proj_2:S^1\times \Sigma\to \Sigma, \quad \ (t,x)\mapsto x\label{eq:projection_fibre}.
\end{align}

For a given $t\in S^1$, consider the inclusion map of each fibre
\begin{align}
  \inclusion_t&:\Sigma\subset S^1\times \Sigma, \quad x\mapsto(t,x)\label{eq:inclusion_fibre}.
\end{align}
 It is useful to list the following identities about these projection and inclusion maps. With $(t,x)\in S^1\times \Sigma$, we have
\begin{align}
 \pi\circ \inclusion_t &= \id_{\Sigma}, \\
 \phi\circ \inclusion_t&=t,\\
 d\pi_{(t,x)}\circ(d\inclusion_t)_x &= \id_{T_x \Sigma}, \\
 d\phi_{(t,x)}\circ (d\inclusion_t)_x &= 0\label{eq:dit-in-kerdphi}
\end{align}
In words, the projection $\pi$ is a left-inverse of the inclusion $\inclusion_t$ for each $t\in S^1$. Consequently, as pullback operators on $k$-forms, we have
\begin{align}
    \inclusion_t^*\pi^* &= \id_{\Omega^k(\Sigma)}, \label{eq:inclusion-pi-is-id} &    \inclusion_t^*\phi^*&=0.
\end{align}
It is immediately recognised that the pre-images of the projection $\phi$ onto the circle are copies of the fibre $\Sigma$, namely $\inclusion_t(\Sigma)=\phi^{-1}(t)$ for all $t\in S^1$.
\begin{lemma}
\label{lem:kernelimage}
For any point $(t,x)\in S^1\times \Sigma$, 
$(d\inclusion_t)_x(T_x \Sigma)=\ker d\phi_{(t,x)}$.
\end{lemma}
\begin{proof}
That $(d\inclusion_t)_x(T_x \Sigma)\subseteq \ker d\phi_x$ is a consequence of identity in \cref{eq:dit-in-kerdphi}. The subspace $\ker d\phi_{(t,x)}$ has co-dimension one since $\phi$ is nowhere critical. Since the inclusion map $\inclusion_t$ is injective, the co-dimension of $(d\inclusion_t)_x(T_x \Sigma)$ is also one, thus both subspaces coincide.
\end{proof}

As a closed one-form on $S^1\times\Sigma$, $d\phi$ represents a basis class of first cohomology. We let $\int_{\Gamma_T} d\phi=2\pi$ where $\Gamma_T$ is a toroidal loop, representing the generator of first homology of the $S^1$ factor.

\subsection{Porting non-autonomous Hamiltonian systems into cosymplectic manifolds}
\label{sec:nonAutoIsCosymplectic}
A modern treatment of Hamiltonian systems generalises the canonical cotangent bundle to consider instead a $2n$-manifold $\Sigma$ together with a symplectic form $\omega_o\in \Omega^2(\Sigma)$. An (autonomous) Hamiltonian system is the dynamics of a Hamiltonian vector field $X_{H_o}$ on $\Sigma$ induced by a (single) Hamiltonian function $H_o :\Sigma \to \R$ satisfying 
  \begin{align*}
    \intprod_{X_{H_o}} \omega_o = -dH_o.
  \end{align*}

  A non-autonomous Hamiltonian system in this setting is a smooth family of Hamiltonian functions $H_t:\Sigma\to\R$ inducing the non-autonomous Hamiltonian vector field $X_t$ on $\Sigma$, via $\intprod_{X_t}\omega_o=-dH_t$. Here, the time-parameter takes values on the circle, $t\in S^1$. The symplectic form $\omega_o\in\Omega^2(\Sigma)$ is fixed (time-independent). A Dirichlet boundary condition is imposed
  \begin{align*}
      \chipboundary^* dH_t=0
  \end{align*}
for all $t\in S^1$, where $\chipboundary:\partial \Sigma \subset \Sigma$. This guarantees that the flow exists for all time (complete flow) when $\Sigma$ is compact.

   It is natural to view the dynamics of our non-autonomous Hamiltonian system as taking place on the $2n+1$ dimensional trivial bundle $S^1\times\Sigma$. The smooth family of Hamiltonian functions $H_t$ converts to the smooth function $H:S^1\times\Sigma\to \R$, $(t,x)\mapsto H_t(x)$. The goal is to construct an evolution vector field $E_H = \partial_t + X_H$ on $S^1\times\Sigma$ analogously to \cref{eq:canonical_evolution-field}. 

\subsubsection{Time component of the evolution vector field}
  \label{sec:time-reeb}
The time vector field $\partial_t\in \Gamma(T(S^1\times\Sigma))$ is uniquely characterised by
\begin{align*}
       d\pi\circ \partial_t &= 0, &
        d\phi(\partial_t) &= 1.
\end{align*}
Defining the closed two-form $\omega = \pi^*\omega_o\in\Omega^2(S^1\times\Sigma)$, the statement $d\pi \circ \partial_t=0$ is equivalent to $\iota_{\partial_t} \omega = 0$ by the non-degeneracy of $\omega_o\in\Omega^2(\Sigma)$. The time vector field $\partial_t\in\Gamma(T(S^1\times\Sigma))$ is thus uniquely determined by
  \begin{align*}
  \partial_t&\in \ker \omega, &
  d\phi(\partial_t)&=1.
  \end{align*}
  Since $\omega^n\wedge d\phi\neq 0$, the tuple $(S^1\times\Sigma,\omega,d\phi)$ forms what is known as a \emph{cosymplectic manifold} \cite{blairContactManifoldsRiemannian1976,chineaTopologyCosymplecticManifolds1993} whose \emph{Reeb field} is precisely $\partial_t$.
\begin{definition}\label{def:cosymplectic-manifold}
The triple $(M,\beta,\eta)$, where $M$ is a $2n+1$-manifold with boundary, $\beta$ a closed two-form on $M$, and $\eta$ a closed one-form on $M$, is said to be a cosymplectic manifold if $\beta^{n} \wedge \eta$ is a non-vanishing top form and $\beta$ satisfies the Dirichlet boundary condition $\boundary^*\beta = 0$ with $\boundary:\partial M\to M$. The unique vector field $R$ such that
\begin{align*}
R&\in \ker \beta, &
\eta(R) &= 1    
\end{align*}
is called the Reeb field of $(M,\beta,\eta)$. 
\end{definition}

%\comment{Of course, $\lie_R\beta=\intprod_R d\beta+d\intprod_R\beta = 0$ and $\lie_R\eta = \intprod_Rd\eta+d\intprod_R\eta = d1=0$, so the flow of the Reeb field preserves the cosymplectic manifold.}

  %
  The integral curves of $\partial_t$ are excessively simple in the sense that they fix points on $\Sigma$ and translate them along the circle direction. In particular, all integral curves of $\partial_t$ are recurrent with the circle's period, or in other words, its flow $\Phi_s^{\partial_t}$ has the property that $\Phi^{\partial_t}_{2\pi}=\id_M$. The vector field $\partial_t$ gives rise to an Ehresmann connection on the trivial fibre bundle $S^1\times \Sigma$. While viewing $\partial_t$ as an Ehresmann connection is excessive in the present setting, we show the utility of this perspective for non-trivial fibre bundles in \cref{sec:MappingToriMCG}.

  \subsubsection{Hamiltonian component of the evolution vector field}
  
  A characterisation of the vector field $X_H$ is that it is tangent to every copy of $\Sigma$ in the product $S^1\times\Sigma$, and that it agrees with the vector field $X_t$ on the specific copy where the last coordinate is equal to $t\in S^1$. This means that for a given time coordinate $t\in S^1$, $X_t$ and $X_H$ are related by the push-forward of the inclusion map $\inclusion_t:\Sigma\to S^1\times\Sigma$ of \cref{eq:inclusion_fibre},
  \begin{align}
  \label{eq:hamiltonian-part}
    d\inclusion_t\circ X_t = X_H\circ \inclusion_t.
  \end{align}
  It is noted that $X_H$ automatically belongs to $\ker d\phi$.
  Rather than through the family of Hamiltonian fields $\{X_t\}_{t\in S^1}$, we wish to characterise $X_H$ through the function $H$. For this, the following result applies.
  \begin{lemma}
  \label{lem:hamiltonian-part}
    The vector field $X_H\in \ker d\phi$ in \cref{eq:hamiltonian-part} is determined by 
    \begin{align*}
      \intprod_{X_H} \omega &= -dH + \partial_t H d\phi.
    \end{align*}
  \end{lemma}

  \begin{proof}
    Indeed, by naturality of the interior product, we have $\inclusion_t^* \intprod_{X_H}\omega = \intprod_{X_t}\omega_o$ and $\inclusion_t^*H = H\circ \inclusion_t =H_t$. This implies $\inclusion_t^* dH=d\inclusion_t^*H=dH_t$ for all $t\in S^1$, so both sides agree when restricted to copies of $\Sigma$, namely on $\ker d\phi$. The expression also holds on $\ker\omega=\Span\{ \partial_t\}$; the left-hand side immediately evaluates to zero, and the extra term on the right-hand side makes the latter vanish too,
    \begin{align*}
      -dH(\partial_t)+\partial_t H \cancelto{1}{d\phi(\partial_t)} 
      = -\partial_t H+\partial_t H
      = 0.
    \end{align*}
  \end{proof}

\subsubsection{The evolution vector field and the related cosymplectic manifold}
  The construction of the evolution vector field is summarised in the following proposition.
%   \begin{proposition}
%   \label{prop:non-autonomous-hamiltonian-systems}
%     A non-autonomous Hamiltonian system on a symplectic manifold $(\Sigma,\omega_o)$ induced by a smooth family of Hamiltonian functions $H_t:\Sigma\to \R$ for $t\in S^1$ has a (autonomous) flow on the trivial bundle over the circle $S^1\times\Sigma$ generated by the smooth function $H:S^1\times \Sigma\to \R$, $(t,x)\mapsto H_t(x)$ given by the evolution vector field 
% \begin{align*}
%     E_H = \partial_t + X_H,
% \end{align*}
% composed of the time vector $\partial_t$ and the Hamiltonian part $X_H$ (uniquely) determined by
% \begin{align*}
%     \partial_t&\in\ker\omega, &  d\phi(\partial_t)&=1 \\
%     X_H &\in\ker d\phi, & \intprod_{X_H} \omega &= -dH + \partial_t H d\phi,
%   \end{align*}
%   where $\omega=\pi^*\omega_o$ and $\phi:S^1\times \Sigma\to S^1$, $\pi:S^1\times \Sigma\to\Sigma$ are the natural projections.
% \end{proposition}

\begin{proposition}\label{prop:non-autonomous-hamiltonian-systems}
    Let $(\Sigma,\omega_o,\{H_t\}_{t\in S^1})$ be a non-autonomous Hamiltonian system. The evolution vector field $E_H$ on $S^1\times\Sigma$ associated to the family of Hamiltonian fields $\{X_t\}_{t\in S^1}$ is given by
\begin{align*}
    E_H = \partial_t + X_H,
\end{align*}
where the time component $\partial_t$ and the Hamiltonian component $X_H$ are (uniquely) determined by
\begin{align*}
    \partial_t&\in\ker\omega, &  d\phi(\partial_t)&=1 \\
    X_H &\in\ker d\phi, & \intprod_{X_H} \omega &= -dH + \partial_t H d\phi,
  \end{align*}
  where $H:S^1\times \Sigma\to \R$ is the smooth function $(t,x)\mapsto H_t(x)$, $\omega=\pi^*\omega_o$, and $\phi:S^1\times \Sigma\to S^1$, $\pi:S^1\times \Sigma\to\Sigma$ are the natural projections.
\end{proposition}

\begin{remark}
\label{rem:cosymplectic-beta}
In \cref{sec:time-reeb}, we identified a cosymplectic manifold $(S^1\times\Sigma,\omega,d\phi)$ such that $\partial_t$ was its Reeb field. By modifying the cosymplectic system, the Reeb field may be adjusted to become the evolution vector field $E_H$. Indeed, take $\beta_o=\omega - dH\wedge d\phi$ and note that, as above, $\beta_o^n\wedge d\phi=\omega^n\wedge d\phi \neq 0$ so that $(S^1\times\Sigma,\beta_o,d\phi)$ is a cosymplectic manifold. The Reeb field $R_o$ of this cosymplectic structure satisfies $R_o\in \ker\beta_o$ and $d\phi(R_o) = 1$. A calculation reveals that $R_o = E_H$. Therefore, one can always view the evolution field $E_H$ for a given $H$ as being the (unique) Reeb field of $\beta_o$.
  \end{remark}
  
% Cross-checking
% \begin{align*}
% \intprod_{\partial_t+X_H}\omega - \intprod_{\partial_t+X_H}dH\wedge d\phi 
% = 
% \cancel{\intprod_{\partial_t}\omega} + \intprod_{X_H}\omega - dH(\partial_t)d\phi + dH d\phi(\partial_t) - \cancel{dH(X_H)}d\phi + dH\cancel{d\phi(X_H)} \\
% = \intprod_{X_H}\omega +dH - \partial_t H d\phi =0
% \end{align*}
% and 
% \begin{align*}
%     d\phi(\partial_t+X_H)=d\phi(\partial_t)=1
% \end{align*}

% \begin{remark}
%        It is important to note in this construction that the time-like part $\partial_t$ generates a simple translation between copies of $\Sigma$ along the circle direction and that the two-form $\omega$ restricted to each $\Sigma$ is a constant symplectic form, $\inclusion_t^*\omega =\inclusion_t^*\pi^*\omega_o=\omega_o$ for all $t\in S^1$. Only the Hamiltonian function $\inclusion_t^*H=H_t$ is changing as a function of the circle value $t\in S^1$.
% \end{remark}

\cref{prop:non-autonomous-hamiltonian-systems,rem:cosymplectic-beta} show how to port the dynamics of a non-autonomous Hamiltonian system from the symplectic manifold $(\Sigma,\omega_o)$ into the cosymplectic manifold $(\Sigma\times S^1,\beta_o,d\phi)$ whose Reeb field is an evolution vector field. This motivates the following definition.

\begin{definition}
\label{def:cosymplectic-hamiltonian-system}
A cosymplectic manifold $(M,\beta,\eta)$ is said to be represented by the non-autonomous Hamiltonian system $(\Sigma,\omega_o,\{H_t\}_{t \in S^1})$ if there exists a diffeomorphism $\Psi: S^1\times \Sigma \to M$ such that $\Psi^*\eta=d\phi$ and
\begin{equation*}
\Psi^*\beta = \omega - dH \wedge d\phi 
\end{equation*}
where $\omega = \pi^*\omega_o$ and $H \in C^{\infty}(S^1 \times \Sigma)$ is given by $H(t,x) = H_t(x)$. 
\end{definition}

\begin{remark}
\label{rem:beta_reeb}
    The representation entails a decomposition of $\beta=\Omega - d\mathcal{H}\wedge\eta$ where $\Omega=\Upsilon^*\omega=\Upsilon^*\pi^*\omega_o$ and $\mathcal{H}=\Upsilon^*H$ with $\Upsilon=\Psi^{-1}$. Consequently, the Reeb field $R$ of $(M,\beta,\eta)$ also admits a decomposition $R=\timelike+\hamlike$ where $\timelike=\Psi_*\partial_t$ and $\hamlike=\Psi_*X_H$ neatly characterised by
    \begin{align*}
        \timelike &\in\ker \Omega, &
        \eta(\timelike)&=1,\\
        \hamlike&\in\ker \eta,&
        \intprod_{\hamlike}\Omega&=-d\mathcal{H}+d\mathcal{H}(\timelike)\eta.
    \end{align*}
% Indeed, 
% \begin{align*}
% \eta(\timelike)\circ\Psi&=\Psi^*(\eta(\Psi_*\partial_t)) = d\phi(\partial_t)=1, &
% \Psi^*(\intprod_{\timelike}\Omega) &= \intprod_{\partial_t}\omega
%      = 0.
% \end{align*}
% Also, $\Psi^*(\intprod_{\hamlike}\Omega + d\mathcal{H}-d\mathcal{H}(\timelike)\eta)
% = \intprod_{X_H}\omega + dH - dH(\partial_t)d\phi = 0$. 

The flow of $\timelike$, denoted $\Phi^{\timelike}_s$, is $\Psi$-related to the flow of $\partial_t$, in the sense that $\Psi\circ\Phi_s^{\partial_t}=\Phi_s^{\timelike}\circ\Psi$. Hence, $\Phi_{2\pi}^{\timelike}=\id_M$ and the field-lines of $\timelike$ are consequently recurrent with the circle's period. The vector field $\timelike$ is called the \emph{time-like part} of $R$ and the vector field $\hamlike$ is called the \emph{Hamiltonian part} of $R$. In \cref{sec:MappingToriMCG}, the time-like part is viewed as an Ehresmann connection producing a monodromy map of the fibre bundle homotopic to the identity.
\end{remark}

We will now consider an important class of examples that enter in the proof of our main result, \cref{thm:moserForPlanarSigma}.

\begin{example}
\label{sec:bundle-isomorphism}
An important class of examples fitting \cref{def:cosymplectic-hamiltonian-system} is generated by diffeomorphisms $\Psi : M \to M$ on the trivial bundle $M = S^1 \times \Sigma$ satisfying $\phi \circ \Psi = \phi$. These are known as \emph{bundle automorphisms}.
We begin with a non-autonomous Hamiltonian system on $(\Sigma,\omega_o)$ with time-dependent Hamiltonian $H \in C^{\infty}(S^1 \times \Sigma)$ and use \cref{prop:non-autonomous-hamiltonian-systems} to port this system into the cosymplectic manifold $(S^1\times\Sigma,\beta_o,d\phi)$ where, as in \cref{rem:cosymplectic-beta}, $\beta_o = \omega-dH\wedge d\phi$ with $\omega=\pi^*\omega_o$. Then, setting $\Upsilon = \Psi^{-1}$, the cosymplectic manifold
\begin{equation*}
(S^1\times\Sigma,\beta,d\phi) = (S^1\times\Sigma,\Upsilon^*\beta_o,d\phi)     
\end{equation*}
is represented by, and therefore equivalent to, the dynamics of the non-autonomous system $(\Sigma,\omega_o,H)$. 

On the other hand, the Reeb field of $(S^1\times\Sigma,\beta,d\phi)$ induces dynamics on $\Sigma$ via the family of functions $\mathcal{H}_t=\inclusion_t^*\mathcal{H}$ where $\mathcal{H}=\Upsilon^*H$ and the family of symplectic two-forms $\omega_t=\inclusion_t^*\beta$ yielding the family of Hamiltonian vector fields $\mathcal{X}_t$ satisfying
\begin{equation*}
\intprod_{\mathcal{X}_t}\omega_t = -d\mathcal{H}_t,
\end{equation*}
which we consider as a natural generalisation of a non-autonomous Hamiltonian system on $\Sigma$. In local \emph{non-canonical} coordinates $z=(z^1,\ldots,z^{2n})\in \Sigma$, such a system would be represented by the differential equations
\begin{align*}
    \frac{dz^i}{dt}(t) = \Pi^{ij}(t,z(t))\frac{\partial \mathcal{H}}{\partial z^j}(t,z(t))
\end{align*}
where $\Pi^{ij}$ is the \emph{Poisson tensor} defined as the inverse of the symplectic tensor~\cite{goldstein}; if $(\omega_t)_{ij}:= \omega_t(\partial_{z_i},\partial_{z_j})$ then $(\omega_t)_{ij}(t,z)\Pi^{jk}(t,z)=\delta_i^k$.
\end{example}
In the next section, we consider the slightly more general notion of \emph{a non-autonomous locally-Hamiltonian system}.

\subsection{Non-autonomous locally-Hamiltonian systems}
\label{sec:locally-hamiltonian-systems}
The essential idea of a locally-Hamiltonian system is that the usual exact one-forms $dH_t$ derived from the Hamiltonian function can be replaced by closed one-forms $\nu_t$ on $\Sigma$. This generalisation is elementary from the perspective that closed one-forms are locally exact.

The triple $(\Sigma,\omega_o,\{\nu_t\}_{t \in S^1})$ is said to be a non-autonomous locally-Hamiltonian system if $\Sigma$ is a manifold with (possibly empty) boundary, $\omega_o$ is a symplectic two-form on $\Sigma$, and $\{\nu_t\}_{t \in S^1}$ is a smooth family of closed one-forms on $\Sigma$, namely $d\nu_t=0$ for all $t\in S^1$. The smooth adjective means that the assignment $S^1 \times \Sigma \to T^*\Sigma$, $(t,x)\mapsto \nu_t\big|_x$ is smooth. A Dirichlet boundary condition is imposed
\begin{equation*}%\label{eq:dirichletconditiononchi}
\chipboundary^*\nu_t = 0
\end{equation*}
for all $t \in S^1$, where $\chipboundary : \partial \Sigma \subset \Sigma$ is the inclusion. The dynamics of $(\Sigma,\omega_o,\{\nu_t\}_{t\in S^1})$ on $\Sigma$ are given by the flow of the non-autonomous vector field $X_t$ defined by
\begin{equation*}
    i_{X_t} \omega_o = \nu_t
\end{equation*}
which is necessarily tangent to $\partial \Sigma$ by the Dirichlet boundary condition on $\nu_t$ and makes the flow complete when $\Sigma$ is compact.

An important part of this paper is determining when a non-autonomous locally-Hamiltonian system is in fact a (globally) non-autonomous Hamiltonian system. The other direction is always true: if one begins with a (globally) non-autonomous Hamiltonian system $(\Sigma,\omega_o,\{H_t\}_{t \in S^1})$, then the family of closed Dirichlet one-forms $\{\nu_t=-dH_t\}_{t \in S^1}$ is smooth, and $(\Sigma,\omega_o,\{\nu_t=-dH_t\}_{t \in S^1})$ is a non-autonomous locally-Hamiltonian system. However, beginning from a non-autonomous locally Hamiltonian system $(\Sigma,\omega_o,\{\nu_t\}_{t\in S^1})$, if there existed an arbitrary family $\{\tilde{H}_t\}_{t \in S^1}$ of functions satisfying $-d\tilde{H}_t = \nu_t$ for each $t\in S^1$, it is not immediate that it may be replaced by a \textbf{smooth} family $\{H_t\}_{t \in S^1}$ satisfying $-dH_t = \nu_t$. 
The following lemma, whose proof is contained in \cref{sec:exact-is-exact}, resolves this. A similar result was shown in \cite{gotay-lachov-weinstein-1983} for $t\in \R^m$.
\begin{lemma}\label{lem:exactfamily}
Let $(\Sigma,\omega_o,\{\nu_t\}_{t \in S^1})$ be a non-autonomous locally-Hamiltonian system, and suppose that $\nu_t$ is exact for each $t \in S^1$. Then, there exists a smooth family $\{H_t\}_{t\in S^1}$ of functions such that
\begin{equation*}
\nu_t = -dH_t.    
\end{equation*}
In particular, $(\Sigma,\omega_o,\{H_t\}_{t \in S^1})$ is a non-autonomous Hamiltonian system.
\end{lemma}

In order to deal more directly with non-autonomous locally-Hamiltonian systems, we continue the analogy with the evolution vector field introduced for non-autonomous (globally) Hamiltonian systems. In the local situation, we can still port the dynamics into a cosymplectic manifold similarly to \cref{prop:non-autonomous-hamiltonian-systems}.

\begin{proposition}
  \label{prop:symplectic-systems}
    Let $(\Sigma,\omega_o,\{\nu_t\}_{t \in S^1})$ be a non-autonomous locally-Hamiltonian system. The evolution vector field $E_H$ on $S^1\times \Sigma$ associated to the family of locally-Hamiltonian fields $\{X_t\}_{t\in S^1}$ is given by 
\begin{align*}
    E_\nu = \partial_t + X_\nu,
\end{align*}
where the time component $\partial_t$ and the locally-Hamiltonian component $X_\nu$ are (uniquely) determined by
\begin{align*}
    \partial_t & \in\ker\omega, &  d\phi(\partial_t)&=1 \\
    X_\nu &\in\ker d\phi, & \intprod_{X_\nu} \omega &= \nu-\nu(\partial_t)d\phi,
  \end{align*}
  where $\nu \in \Omega^1(S^1\times \Sigma)$ is such that $\inclusion_t^*\nu = \nu_t$ for all $t\in S^1$, $\omega=\pi^*\omega_o$ and $\phi:S^1\times \Sigma\to S^1$, $\pi:S^1\times \Sigma\to\Sigma$ are the natural projections.
\end{proposition}

By virtue of \cref{prop:symplectic-systems}, we obtain a generalisation of \cref{def:cosymplectic-hamiltonian-system}.

\begin{definition}\label{def:cosymplectic-symplectic-system}
A cosymplectic manifold $(M,\beta,\eta)$ is said to be represented by the non-autonomous locally-Hamiltonian system $(\Sigma,\omega_o,\{\nu_t\}_{t\in S^1})$ if there exists a diffeomorphism $\Psi: S^1\times \Sigma \to M$ such that $\Psi^*\eta=d\phi$ and 
\begin{align*}
    \Psi^*\beta=\omega + \nu \wedge d\phi
\end{align*}
where $\omega=\pi^*\omega_o$ and $\nu\in\Omega^1(S^1\times\Sigma)$ satisfies $\inclusion_t^*\nu=\nu_t$ for all $t\in S^1$.
\end{definition}

\begin{remark}
The restriction of $\nu$ to a copy of $\Sigma$ agrees with $\nu_t$ at the corresponding angle $t\in S^1$, which means that $\nu$ is determined by the smooth family of one-forms $\{\nu_t\}_{t\in S^1}$ up to a term of the form $f d\phi$ for some function $f$. An explicit representative is, for example,
\begin{align*}
    \hat{\nu}|_{(t,x)}=(\hat{\nu}_t)|_x\circ d\pi_{(t,x)}=(\pi^*\hat{\nu}_t)|_{(t,x)},
\end{align*}
which, incidentally, is Dirichlet on $S^1\times\Sigma$. Indeed, for each $(t,x) \in \partial (S^1\times\Sigma)=S^1\times\partial\Sigma$,
\begin{align*}
(\boundary^*\hat{\nu})|_{(t,x)} 
&= \hat{\nu}|_{(t,x)} \circ d\boundary_{(t,x)} 
= (\hat{\nu}_t)|_x\circ d\pi_{(t,x)}\circ d\boundary_{(t,x)} 
= (\hat{\nu}_t)|_x\circ d(\pi\circ\boundary)_{(t,x)}
= (\hat{\nu}_t)|_x\circ d(\chipboundary \circ \pi_\partial)_{(t,x)}\\
&= (\hat{\nu}_t)|_x \circ d\chipboundary_{x} \circ d{\pi_\partial}_{(t,x)}
= (\chipboundary^*\hat{\nu}_t)|_x \circ d{\pi_\partial}_{(t,x)} 
= 0,
\end{align*}
where $\boundary:\partial(S^1\times\Sigma)\to S^1\times \Sigma$ and $\chipboundary:\partial \Sigma\to\Sigma$ are the inclusions and $\pi_\partial:\partial(S^1\times\Sigma)\to \partial \Sigma$ is the projection, so that $\pi\circ \boundary=\chipboundary\circ\pi_\partial$.
\end{remark}

Representative one-forms $\nu$ in \cref{prop:symplectic-systems} are closed on each copy of $\Sigma$, i.e. $\inclusion_t^*d\nu=d\inclusion_t\nu=d\nu_t=0$. Unfortunately, there may not exist a choice that is closed on all of $S^1\times \Sigma$. This is a serious obstacle to reconciling the notion of locally-Hamiltonian to its global counter-part. The following lemma precisely describes the issue.
\begin{lemma}\label{lem:obstacle}
Let $\{\nu_t\}_{t \in S^1}$ be a smooth family of closed one-forms on a manifold $\Sigma$ with boundary. Then, the following are equivalent.
\begin{enumerate}
    \item There exists a closed one-form $\nu$ on $S^1\times \Sigma$ such that $\inclusion_t^*\nu=\nu_t$ for each $t\in S^1$;
    \item The cohomology classes $[\nu_t]$ are constant.
\end{enumerate}
\end{lemma}
\begin{proof}
Assume that the one-form $\nu$ on $S^1\times \Sigma$ such that $\inclusion_t^*\nu = \inclusion_t^*\nu = \nu_t$ for each $t\in S^1$ is closed. If one is given a smooth closed curve $\gamma : [0,1] \to \Sigma$, then for any $s,t \in S^1$, the curves $\gamma_s = \inclusion_s \circ \gamma$ and $\gamma_t = \inclusion_t \circ \gamma$ are cohomologous on $S^1\times \Sigma$ and thus
\begin{equation*}
\int_\gamma \nu_s  = \int_{\gamma_s} \nu = \int_{\gamma_t} \nu = \int_\gamma \nu_t 
\end{equation*}
which implies that the cohomology classes $[\nu_s]$ and $[\nu_t]$ coincide. 

Now assume that the smooth family $\{\nu_t\}_{t\in S^1}$ of closed one-forms has constant cohomology class. Then, fixing $o\in S^1$, $\{\nu_t-\nu_o\}_{t\in S^1}$ is a smooth family of exact one-forms. Therefore, by \cref{lem:exactfamily}, there exists a smooth family $\{H_t\}_{t \in S^1}$ of functions such that $\nu_t-\nu_o = -dH_t$ for each $t \in S^1$. That is, $\nu_t=-dH_t+\nu_o$. Defining the smooth function $H$ on $S^1\times \Sigma$ via $H(t,x)=H_t(x)$, the one-form
\begin{align*}
    \nu = -dH + \pi^*\nu_o
\end{align*}
is manifestly closed and, by virtue of \cref{eq:inclusion-pi-is-id}, such that $\inclusion_t^*\nu = -d\inclusion_t^*H+\inclusion_t^*\pi^*\nu_o = -dH_t+\nu_o=\nu_t$ for all $t\in S^1$ as required.
\end{proof}

The obstacle highlighted by \cref{lem:obstacle} is subtle. On a contractible neighbourhood $U$ of any point in $\Sigma$, the closed family of $\{\nu_t\}_{t\in S^1}$ are exact, which means that a global function $h : S^1 \times U \to \R$ exists such that $\nu_t|_U=-dh_t$ for each $t\in S^1$. The temptation is to set $\tilde{\nu}=-dh$ on the extended neighbourhood $S^1\times U$ as a candidate to close the gap between locally-Hamiltonian and (globally) Hamiltonian. The strategy of covering $\Sigma$ with contractible neighbourhoods in an attempt to glue $\tilde{\nu}$ into a globally closed one-form on $S^1 \times \Sigma$ will however fail unless the conditions of \cref{lem:obstacle} are met.

It is easy to construct non-examples of \cref{lem:obstacle}, including cases with non-empty boundary. Indeed, if $(\Sigma,\omega_o)$ is a symplectic manifold with boundary admitting a closed Dirichlet one-form $\nu_o$ with non-trivial de Rham cohomology class (see \cref{eg:varying-cohomology-classes}), a family $\{\nu_t\}_{t \in S^1}$ of Dirichlet closed one-forms with varying cohomology classes on $\Sigma$ can be generated. For instance, taking $\nu_t = (2 + \sin t)\nu_o$, the cohomology class varies in $t$ and therefore there cannot be a closed $\nu$ on $S^1\times\Sigma$ such that $\inclusion_t^*\nu=\nu_t$. 

Remarkably though, when $\Sigma$ is a compact regular domain of $\mathbb{R}^m$ where $m=2n$ (even with non-zero first de Rham cohomology groups), such non-examples cannot be found. Briefly, closed Dirichlet one-forms on such $\Sigma$ are always exact, in which case the notion of non-autonomous Hamiltonian system automatically collapses to its global counter-part.
\begin{proposition}\label{prop:boundarycohomology}
 Let $\Sigma \subset \mathbb{R}^m$ be a compact regular domain. The boundary inclusion map $\boundary: \partial \Sigma \subset \Sigma$ is injective on the first (absolute) cohomology as a pull-back map
    \begin{align*}
        \boundary^*:H_{\text{dR}}^1(\Sigma)\to H_{\text{dR}}^1(\partial \Sigma).
    \end{align*}
In particular, if $\nu_o$ is a closed one-form on $\Sigma$ such that $\boundary^*\nu_o = 0$, then $\nu_o = -dH_o$ for some $H_o\in C^\infty(\Sigma)$.
\end{proposition}

\begin{proof}
We follow Cantarella's sketch \cite{Cantarella_DeTurck_Gluck_2002} for $m=3$ adapted to arbitrary $m$ with the details filled-in for completeness. Specifically, we first work with the first singular homology of $\Sigma$ via a Mayer-Vietoris sequence as follows. Recall that $\Sigma$ has a normal collar, namely a smooth embedding $\sigma : \partial \Sigma \times [0,1) \to \Sigma$. Let $C$ be the image of $\sigma$ and consider the open set
\begin{equation*}
N =  \R^m\backslash (\Sigma \setminus C) = (\R^m \setminus \Sigma)\cup C
\end{equation*}
of $\R^m$. Then, the interiors of $\Sigma$ and $N$ in $\R^m$ cover $\R^m$, so we may apply the Mayer-Vietoris sequence to $\Sigma$ and $N$. Part of the sequence reads
\begin{equation*}
... \to H_2(\R^m) \to H_1(\Sigma \cap N) \to H_1(\Sigma) \oplus H_1(N) \to H_1(\R^m) \to ...
\end{equation*}
which, because $C \cong \partial \Sigma \times [0,1)$ deformation retracts to $\partial \Sigma \subset C$, reduces to
\begin{equation*}
0 \to H_1(\partial \Sigma) \to H_1(\Sigma) \oplus H_1(N) \to 0.
\end{equation*}
The map $H_1(\partial \Sigma) \to H_1(\Sigma) \oplus H_1(N)$ is $(\boundary_*,k_*)$ where $\boundary : \partial \Sigma \subset \Sigma$ and $k : \partial \Sigma \subset N$ are the inclusions. By exactness of the sequence, $(\boundary_*,k_*)$ is a surjection. In particular, the map $\boundary_* : H_1(\partial \Sigma) \to H_1(\Sigma)$ is a surjection. Dually, by de Rham's Theorem (for manifolds with boundary), the pull-back map $\jmath^* : H_{\text{dR}}^1(\Sigma)\to H_{\text{dR}}^1(\partial \Sigma)$ is injective.

In particular, if $\nu_o$ is a closed one-form on $\Sigma$ such that $\boundary^*\nu_o = 0$, we have that $[\nu_o] \in  H_{\text{dR}}^1(\Sigma)$ is in the kernel of $\boundary^* : H_{\text{dR}}^1(\Sigma)\to H_{\text{dR}}^1(\partial \Sigma)$, and therefore, by injectivity of $\boundary^*$, $[\nu_o] = 0$, and therefore, $\nu_o = -dH$ for some $H \in C^{\infty}(\Sigma)$.
\end{proof}
We can now state and prove that non-autonomous locally-Hamiltonian systems reduce to non-autonomous (globally) Hamiltonian systems on compact regular domains in Euclidean space.
\begin{corollary}\label{prop:symplectic-is-hamiltonian}
Let $\Sigma$ be a compact regular domain of $\mathbb{R}^{2n}$. Then, if $(\Sigma,\omega_o,\{\nu_t\}_{t \in S^1})$ is a non-autonomous locally-Hamiltonian system, there exists a smooth family $\{H_t\}_{t\in S^1}$ of functions such that
\begin{equation*}
\nu_t = -dH_t.    
\end{equation*}
In particular, $(\Sigma,\omega_o,\{H_t\}_{t \in S^1})$ is a non-autonomous Hamiltonian system.
\end{corollary}

\begin{proof}
For each $t \in S^1$, observe that, by definition of non-autonomous locally-Hamiltonian system, the closed one-forms are Dirichlet, 
\begin{equation*}
\chipboundary^*\nu_t = 0.    
\end{equation*}
In particular, by \cref{prop:boundarycohomology}, $\nu_t$ is exact. The result now follows from \cref{lem:exactfamily}.
\end{proof}

\section{Reduction of divergence-free vector fields to non-autonomous 1\texorpdfstring{$\frac{1}{2}$}{.5}D Hamiltonian systems}
\label{sec:divergence-free-fields}
%%

%\subsection{Preliminary remarks}
%The manifold of interest and of relevance to magnetic confinement fusion is a solid torus $M\cong S^1\times D^2$ embedded in Euclidean three-space. 

%The advantage of formulating magnetostatics in terms of differential forms is that their behaviour under algebraic differential manipulations (pullbacks, exterior derivatives, interior products, Lie derivatives, etc..) does not require introducing a Riemannian metric. This means that many properties are structural/topological in nature, rather than geometrical.

%All physical instances of magnetic fields originate from exact magnetic two-forms, \Pfeff{[This is thanks to Perrella's consideration of embedded manifolds and the Dirichlet condition]}.
%
In this section, we will use the previously introduced technology to establish the main result around a sufficient condition on magnetic fields with Poincar\'e sections to admit a non-autonomous Hamiltonian representation. 
%We start with a Dirichlet closed two-form $\beta$ on a compact orientable three-manifold $M$ possibly with boundary.
We start with a vector field $B$ tangent to the boundary of a compact orientable three-manifold $M$ equipped with a volume form $\mu$. We ask that $B$ is divergence-free in the sense that $\beta=\intprod_B\mu$ is a closed two-form. The boundary condition on $\beta$ inherited from $B$ is the Dirichlet boundary condition.
    \begin{lemma}
    \label{lem:boundaryconditions}
    Denoting the natural inclusion of the boundary $\boundary:\partial M\subset M$, the following statements are equivalent. 
        \begin{enumerate}
            \item The magnetic two-form $\beta$ satisfies the \emph{Dirichlet} boundary condition, namely $\boundary^*\beta=0$;
            \item The magnetic field $B$ is tangential to the boundary, namely $B_{\boundary(x)}\in d\boundary_x(T_x\partial M)$ for all $x\in\partial M$.
        \end{enumerate}
    \end{lemma}
    Using any metric, the latter statement amounts to the fact that the magnetic field on the boundary is orthogonal to the normal vector, $\bm{n}\cdot B|_{\partial M}=0$.
    \begin{proof}
    Let $x\in\partial M$ and $w_1=d\boundary_x(v_1),w_2=d\boundary_x(v_2)\in d\boundary_x(T_x\partial M)$ be two linearly independent vectors spanning the subspace of tangential vectors to the boundary at point $\boundary(x)$. Then, $(\boundary^*\beta)_x(v_1,v_2)=\beta_{\boundary(x)}(w_1,w_2)=\mu_{\boundary(x)}(B_{\boundary(x)},w_1,w_2)$. The Dirichlet condition $\boundary^*\beta=0$ implies that $B_{\boundary(x)}$ is linearly dependent with $w_1,w_2$, namely $B_{\boundary(x)}$ is tangential to the boundary, and conversely.
    \end{proof}

\begin{remark}
\cref{prop:boundarycohomology} implies that all tangential magnetic fields $B$ on a compact regular domain $M\subset \R^3$ give rise to an exact two-form $\intprod_B \mu=\beta=d\alpha$. Indeed, let $\nu$ be a closed one-form with Dirichlet boundary condition. Then, $\boundary^*\nu = 0$ so that on cohomology, $\boundary^*[\nu] = [\boundary^*\nu] = 0$. By injectivity of $\boundary^*$, we have that $[\nu] = 0$; that is, $\nu = dH$ for some $H$.  Thus
\begin{equation*}
\int_{M} \nu \wedge \beta = \int_{M} dH \wedge \beta = \int_{M} d(H\beta) = \int_{\partial M}  H\beta = 0.
\end{equation*}
By Poincar\'e-Lefschetz duality, this implies that $\beta$ must be exact in the absolute sense. That is, $\beta = d\alpha$.    
\end{remark}

We restrict our attention to magnetic fields that admit a compact global Poincaré section; namely a compact connected embedded submanifold $\Sigma$ of $M$ with $\partial \Sigma \subset \partial M$, such that every field-line of $B$ transversely intersects $\Sigma$. This limiting but physically motivated assumption is equivalent to the existence of an angle function $\phi:M\to S^1$ such that the $d\phi(B) > 0$ has no zeros (so that, in particular, $B$ has no zeros), as a result of the following lemma whose proof is deferred to \cref{sec:proof-equivalence-anglefunction}.
\begin{lemma}
\label{lem:equivalence-anglefunction}
If $M$ is a compact orientable manifold with boundary and $X$ is a vector field tangent to $\partial M$ with a global Poincar\'e section $\Sigma$, then there exists an angle function $\phi : M \to S^1$ such that $d\phi(X) > 0$ and $\Sigma$ is a level-set of $\phi$.
\end{lemma}
The angle function enables us to view $M$ as the fibre-bundle $(M,S^1,\phi,\Sigma)$ over the circle. This hugely limits its topology. For example, removing a ball from a solid torus will prevent any tangential vector field from admitting a Poincaré section (Hairy Ball Theorem).

For reference, the positive function $d\phi(B)>0$ is called the \emph{toroidal component} of $B$. In this setting, $\tilde{\mu}:=\beta\wedge d\phi$ is a volume form and the tuple $(M,\beta,d\phi)$ forms a cosymplectic manifold. It is straightforward to verify that the Reeb field $R_B$ such that $R_B\in \ker \beta$ and $d\phi(R_B)=1$ is the magnetic field normalised to its toroidal component,
    \begin{align*}
        R_B = \frac{B}{d\phi(B)}.
    \end{align*}
    It is divergence-free with respect to the volume-form $\tilde{\mu}=\beta\wedge d\phi$.
%     Indeed,
% \begin{align*}
%     d\intprod_{\tilde B} \tilde{\mu} = d \intprod_B \frac{d\phi(B)}{d\phi(B)}\mu = d\intprod_B\mu =d\beta=0.
% \end{align*}
The field-lines of the Reeb field and those of the magnetic field coincide. The fact that $R_B$ is normalised to one in the circle direction, $d\phi(R)=1$, makes the correspondence between the circle value of an integral curve (field-line) and the time-parameter simple. It is in this sense that the circle direction will constitute the ``$\frac{1}{2}$'' dimension of the equivalent Hamiltonian system.

We restrict to the case where $M = S^1 \times \Sigma$, where $\Sigma$ is an orientable connected compact two-manifold possibly with boundary, and $\phi : M \to S^1$ is the natural projection. The standard case of a solid torus embedded in real three-space, $\hat{M}\subseteq \R^3$ and $\hat{M}\cong S^1\times D^2$, is an important example covered by our analysis. If one has a divergence-free vector field $\hat{B}$ on a solid torus $\hat{M} \subset \mathbb{R}^3$, if $F : \hat{M} \to S^1 \times D^2$ is a diffeomorphism such that $F^*\phi=\hat{\phi}$ is an angle function with $d\hat{\phi}(\hat{B}) > 0$, then the algebraic results will be based on the two-form $\beta$ on $S^1\times D^2$ such that $F^*\beta=\hat{\beta}$, and the conclusions carry over to the original two-form $\beta$ diffeomorphically. Another application of our results is to the hollow torus $\tilde{M}\cong S^1\times A_{1,2}$ embedded in real three-space, where $A_{1,2}$ is an annulus (or cylinder), see \cref{eq:annulus}. 

In fact, as soon as the Poincar\'e section $\Sigma$ is diffeomorphic to a disk or an annulus, then $\hat{M}$ is necessarily a trivial bundle over the circle, as concluded from \cref{thm:MappingClassGroupIsomorphisms} presented in \cref{sec:MappingToriMCG}.
\begin{corollary}
\label{cor:disk-annulus-trivial}
Let $(\hat{M},S^1,\hat{\phi},\Sigma)$ be a fibre-bundle over the circle with typical fibre $\Sigma$. If $\Sigma$ is diffeomorphic to the disk or annulus, then $(\hat{M},S^1,\hat{\phi},\Sigma)$ is necessarily trivialisable. That is, there exists a diffeomorphism $F : \hat{M} \to S^1 \times \Sigma$ such that $F^*\phi=\hat{\phi}$.
\end{corollary}

\begin{proof}
    By \cref{thm:MappingClassGroupIsomorphisms}, when $\Sigma$ is a disk or an annulus, the mapping class group is trivial. Moreover, from \cref{thm:bundlesAreMCG}, the conjugacy classes of the mapping class group are in one-to-one correspondence with isomorphism classes of fibre bundles with fibre $\Sigma$ (\cref{thm:bundlesAreMCG}). It follows that all fibre bundles with fibre $\Sigma$ an annulus or disk are bundle isomorphic to the trivial bundle.
\end{proof}

Now, we ask whether the cosymplectic manifold $(S^1\times\Sigma,\beta,d\phi)$ can be represented by some non-autonomous Hamiltonian as in \cref{def:cosymplectic-hamiltonian-system,rem:beta_reeb}. Using a general Moser's trick (see \cref{thm:moser}, proven later in this section), and \cref{prop:symplectic-is-hamiltonian}, we obtain the following result in the case where $\Sigma$ is planar; that is, when $\Sigma$ is a compact regular domain in $\mathbb{R}^2$.

\begin{theorem}\label{thm:moserForPlanarSigma} 
Consider the cosymplectic manifold $(S^1\times\Sigma,\beta,d\phi)$ where $\Sigma$ is planar. Then, $(S^1\times\Sigma,\beta,d\phi)$ admits a non-autonomous Hamiltonian representation. In fact, there exists a smooth family of functions $\{H_t\}_{t \in S^1}$ such that $(\Sigma,\inclusion_o^*\beta,\{H_t\}_{t \in S^1})$ is a non-autonomous Hamiltonian system and there exists a bundle automorphism $\Psi$ such that
\begin{equation*}
  \beta = \Omega - d\mathcal{H}\wedge d\phi
\end{equation*}
where $(\Psi^*\mathcal{H})(t,x) = H_t(x)$ and $\Psi^*\Omega=\pi^*\inclusion_o^*\beta$ for an arbitrary (fixed) reference angle on the circle $o\in S^1$.
\end{theorem}

\begin{remark}
We note that if $\Sigma$ is a given two-manifold with boundary such that the product $M = S^1 \times \Sigma$ can be embedded in $\R^3$, then $\Sigma$ is automatically planar. This fact is proven in \cref{sec:planar-if-embedded}.
\end{remark}

%\subsection{Manifolds with disk or annulus cross-sections}

We conclude this section by highlighting the following remarkable conclusion; any volume-preserving vector field on a three-manifold with a global Poincar\'e section diffeomorphic to a disk or an annulus admits a non-autonomous Hamiltonian representation. 
This is a direct application of \cref{thm:moserForPlanarSigma,lem:equivalence-anglefunction,cor:disk-annulus-trivial}, which relies upon the classification of bundles over a circle and the relation to mapping class groups, detailed in \cref{sec:MappingToriMCG}.

The next section revolves around the proof of \cref{thm:moserForPlanarSigma}.

%%% Local Variables:
%%% mode: latex
%%% TeX-master: "main"
%%% End:

%%%%%%%%%%%%%

\subsection{Global procedure applying Moser's trick}
\label{sec:moser-trick}

We consider the family  of top-forms  on $\Sigma$ labelled by $t\in S^1$ 
\begin{align*}
  \omega_t := \inclusion_t^* \beta .
\end{align*}

Because $\beta\wedge d\phi\neq 0$, we have for all linearly independent vector fields $Y_1,Y_2\in \ker d\phi$ that $\intprod_{Y_2}\intprod_{Y_1}(\beta\wedge d\phi)= \beta(Y_1,Y_2) d\phi \neq 0 \iff \beta(Y_1,Y_2)\neq 0$. This means that $\beta$ is non-degenerate on the kernel of $d\phi$, or equivalently that each two-form $\omega_t$ is non-degenerate on $\Sigma$ so that they form a family of symplectic forms. As top-forms, $\omega_t$ satisfy a Dirichlet boundary condition, namely $\chipboundary^*\omega_t=0$, $\forall t\in S^1$ where $\chipboundary:\partial \Sigma\to \Sigma$.

It is well-known in plasma physics that the flux through any cross-section is invariant, provided the magnetic field is tangential to the boundary of the toroidal domain considered. We capture this physical property in the following lemma.
\begin{lemma}
    The magnetic flux through the level sets of $\phi$ is constant, namely $\forall a,b\in S^1$,
    \begin{align*}
        \int\limits_{\phi^{-1}(a)} \beta =  \int\limits_{\phi^{-1}(b)} \beta.
      \end{align*}
\end{lemma}

\begin{proof}
Assume $a\neq b$. Write $a = a_0 + 2\pi\mathbb{Z}$ and $b = b_0 + 2\pi\mathbb{Z}$ where $0 < b_0-a_0 < 2\pi$. Consider the path $\gamma : [0,1] \to S^1$ where $\gamma(t) = a_0 + (b_0-a_0)t + 2\pi\mathbb{Z}$. Then, the map $V :  [0,1]\times \Sigma \to S^1\times\Sigma$ given by 
\begin{equation*}
V(t,x) = (\gamma(t),x)
\end{equation*}
is an embedding of the manifold with corners $ [0,1]\times \Sigma$ into $S^1\times\Sigma$. The image, $V([0,1] \times \Sigma)$, has boundary
\begin{equation*}
\partial V ([0,1]\times\Sigma)=V(\partial ([0,1]\times\Sigma))
=V(\partial [0,1]\times\Sigma)\cup V(\{0\}\times\Sigma)\cup V(\{1\}\times\Sigma) 
= C_{ab}\cup \phi^{-1}(a)\cup \phi^{-1}(b).
\end{equation*}
The surface $C_{ab}=V(\partial [0,1]\times\Sigma)\subseteq \partial( S^1\times\Sigma)=S^1\times\partial\Sigma$ belong to the boundary of the manifold, over which the magnetic two-form pulls back to zero by the Dirichlet boundary condition. With this in mind and because the magnetic two-form is closed, Stokes theorem yields
\begin{align*}
  0 = \int\limits_V d\beta
  = \int\limits_{\phi^{-1}(a)}\beta - \int\limits_{\phi^{-1}(b)}\beta
\end{align*}
where the minus sign comes from the opposite orientation of the level sets $\phi^{-1}(a)$ and $\phi^{-1}(b)$ under $V$.
\end{proof}

\begin{corollary}
\label{lem:cohomologous}
  The area of $\Sigma$ computed via the symplectic forms $\omega_t$ is independent of $t\in S^1$. As closed Dirichlet two-forms, $\omega_t$ thus belong to the same relative cohomology class.
\end{corollary}
\begin{proof}
  Indeed, 
  \begin{align*}
    \int\limits_{\Sigma} \omega_t
    = \int\limits_{\Sigma}\inclusion_t^* \beta
    = \int\limits_{\inclusion_t(\Sigma)}\beta
    =\int\limits_{\phi^{-1}(t)}\beta.
  \end{align*}
\end{proof}

That all two-forms $\omega_t$ are relative cohomologous allows us to relate them by a family of diffeomorphisms. This is essentially the content of Moser's stability theorem~\cite{moser-1965}. However, it is important to note that the family of two-forms $\omega_t$ has a periodic parameter $t\in S^1$. If one naively applies Moser's stability theorem, then one would obtain a smooth family $\{\Psi_t : \Sigma\to \Sigma\}_{t \in [0,1]}$ of diffeomorphisms such that $\Psi_t^*\omega_{\gamma(t)} = \omega_o$ where $o\in S^1$ is fixed and $\gamma : [0,1] \to S^1$ is the quotient projection. However, it is not immediate that the family $\{\Psi_t : \Sigma\to \Sigma\}_{t \in [0,1]}$ can be chosen with $\Psi_0 = \Psi_1$. Therefore, some care is needed to apply Moser's result to ensure that one obtains a well-defined bundle automorphism of $S^1\times \Sigma$. This leads to a \emph{parameterised} version of his theorem. First, a definition to easily state it.

\begin{definition}
  Let $\nu_1,\nu_2\in \Omega^k(S^1\times\Sigma)$ be $k$-forms on $S^1\times\Sigma$. We write
  \begin{align*}
    \nu_1=\nu_2 \text{ mod } S^1 \quad \text{ if } \quad \inclusion_t^*\nu_1=\inclusion_t^*\nu_2, \quad \forall t\in S^1.
  \end{align*}
\end{definition}

\begin{lemma} The statement $\nu_1=\nu_2 \text{ mod } S^1$ is equivalent to 
\begin{align*}
\nu_1(Y_1,\ldots,Y_k) = \nu_2(Y_1,\ldots Y_k),\quad \forall Y_1,\ldots,Y_k\in \ker d\phi.
\end{align*} 
\end{lemma}
\begin{proof}
Fix $t\in S^1$, $x\in \Sigma$ and let $(t,x)\in S^1\times\Sigma$. By \cref{lem:kernelimage}, $(d\inclusion_t)_x$ forms a bijection between $T_x\Sigma$ and $\ker d\phi_{(t,x)}$. So, evaluating the pullback $(\inclusion^*_t\nu_m)_x$ over $T_x\Sigma$ is the same as evaluating $(\nu_m)_{(t,x)}$ on $\ker d\phi_{(t,x)}$. Explicitly, let $v_1,\ldots,v_k\in T_x\Sigma$ and the corresponding $w_1,\ldots,w_k\in \ker d\phi_{(t,x)}$ such that $w_n=(d\inclusion_t)_x(v_n)$, for $n=1,\ldots,k$. Then,
\begin{align*}
    (\inclusion_t^*\nu_m)_x(v_1,\ldots,v_k) = (\nu_m)_{(t,x)}((d\inclusion_t)_x(v_1),\ldots,(d\inclusion_t)_x(v_k))
    = (\nu_m)_{(t,x)}(w_1,\ldots,w_n) , \quad m=1,2.
\end{align*}
\end{proof}

\begin{remark}
\label{rem:modS1}
    Applied to one-forms, $\alpha_1=\alpha_2$ mod $S^1$ is equivalent to $\alpha_1=\alpha_2+ fd\phi$ for some function $f$. In the case of two-forms, $\omega_1=\omega_2$ mod $S^1$ is equivalent to $\omega_1=\omega_2+\chi\wedge d\phi$ for some one-form $\chi$.
\end{remark}

We are now ready to state and prove an important result.
\begin{theorem}
\label{thm:moser}
  There exists a bundle isomorphism $\Psi:S^1\times\Sigma\to S^1\times\Sigma$ with $\Psi^*\phi=\phi$ such that $\Psi^*\beta = \omega \mod S^1$, where $\omega=\pi^*\inclusion_o^* \beta$ for an arbitrary (fixed) reference angle on the circle $o\in S^1$.
\end{theorem}

\begin{proof}
    The procedure is known as Moser's trick. Consider the family of closed two-forms $\beta_s=s\beta+(1-s)\omega$ for $s\in [0,1]$, linearly interpolating between the closed two-forms $\omega$ and $\beta$. We wish to construct a family of diffeomorphisms $\Psi_s$ so that $\Psi_s^*\phi=\phi$ and $\Psi_s^*\beta_s=\omega$ mod $S^1$ for each $s\in[0,1]$. To achieve this, we manufacture a non-autonomous vector field $X_s$ on $S^1\times\Sigma$ and solve the linear differential equation with variable coefficients, $\frac{d}{ds} \Psi_s = X_s\circ\Psi_s$, supplemented with initial condition $\Psi_0=\id_{S^1\times\Sigma}$. The non-autonomous vector field $X_s$ uniquely determines $\Psi_s$, and vice-versa.
    
    The property that $\Psi_s^*\phi=\phi$ is equivalent to the generating vector field $X_s$ belonging to $\ker d\phi$. This can be seen by differentiating the condition with respect to $s$,
    \begin{align*}
        0&=\frac{d}{ds}(\Psi_s^*\phi-\phi)= \Psi_s^* \lie_{X_s}\phi \iff  d\phi(X_s)=0 \iff X_s\in \ker d\phi.
    \end{align*}
   
    The second property $\Psi_s^*\beta_s=\omega$ mod $S^1$ is equivalent to the non-autonomous vector field $X_s\in \ker d\phi$ additionally satisfying
    \begin{align*}
        0&=\frac{d}{ds}(\Psi_s^*\beta_s-\omega) = \Psi_s^*(\partial_s \beta_s + \lie_{X_s}\beta_s)=0 \text{ mod } S^1
        \iff \beta-\omega + d i_{X_s}\beta_s = 0 \text{ mod } S^1,
    \end{align*}    
    where, when $\Psi_s$ preserves $\ker d\phi$ as a result of the first condition, the statement $\Psi_s^*\nu=0$ mod $S^1$ is equivalent to $\nu=0$ mod $S^1$, for $\nu\in \Omega^k(S^1\times\Sigma)$.
Now by \cref{lem:cohomologous},  $\inclusion_t^*(\beta-\omega)$ is relative cohomologous to zero for all $t\in S^1$, so $\beta-\omega = d\alpha$ mod $S^1$ for some Dirichlet-mod $S^1$ one-form $\alpha$ on $S^1\times\Sigma$. 
The requirements on the generating vector field are thus $X_s\in \ker d\phi$ and    $d(\intprod_{X_s}\beta_s + \alpha)=0$ mod $S^1$. Conveniently, it is sufficient that $X_s\in \ker d\phi$ satisfies
    \begin{align*}
        \intprod_{X_s}\beta_s(Y)&=-\alpha(Y),\quad \forall Y\in \ker d\phi,
    \end{align*}
    which determines $X_s$ uniquely because $\beta_s$ is non-degenerate on $\ker d\phi$, for all $s\in[0,1]$.

    Now that the non-autonomous vector field $X_s$ has been prescribed, we verify that the initial-value problem $\frac{d}{ds}\Psi_s = X_s\circ \Psi_s$ with $\Psi_0=\id$ is well-posed. Since $\alpha$ is Dirichlet-mod $S^1$, the vector fields $X_s$ are tangential to the boundary of $S^1\times\Sigma$. 
    Indeed, at any point $x\in \partial (S^1\times\Sigma)\subset S^1\times\Sigma$, let $v\in d\boundary_x(T_x\partial (S^1\times\Sigma))\cap \ker d\phi_x$. Then, $0=-\alpha_x(v) = (\beta_s)_x((X_s)_x,v)$, which means that $(X_s)_x$ are co-linear to $v$, namely tangential to the boundary.
    The manifold $S^1\times\Sigma$ is compact and $X_s$ is tangential to the boundary, which implies that $X_s$ has global non-autonomous flow, $\Psi_s : S^1\times\Sigma\to S^1\times\Sigma$ for all times $t\in[0,1]$. Hence, the desired diffeomorphism $\Psi=\Psi_1$ exists.

\end{proof}

\begin{remark}
\label{rem:chi-nu}
\cref{thm:moser} implies an important decomposition of the magnetic two-form $\beta$. Consider the closed two-form $\Omega$ on $S^1\times\Sigma$ such that $\Psi^*\Omega=\omega=\pi^*\inclusion_o^*\beta$ where $\Psi$ is the diffeomorphism of \cref{thm:moser}. Because $\Psi$ preserves the level sets of $\phi$, the two-form $\Omega$ is Dirichlet, $\boundary^*\Omega=0$. By virtue of \cref{thm:moser}, the magnetic two-form $\beta$ matches $\Omega$ on the level sets of $\phi$, namely $\beta=\Omega$ mod $S^1$ (or $\inclusion_t^*\beta=\inclusion_t^*\Omega=\omega_t$, $\forall t\in S^1$) which means, by \cref{rem:modS1},
\begin{align*}\label{eq:decomposition}
    \beta = \Omega + \chi\wedge d\phi
\end{align*}
for some one-form $\chi$. 

The one-form $\chi$ necessarily satisfies additional properties. Because $\beta$, $\Omega$ and $d\phi$ are all closed, the one-form $\chi$ satisfies
\begin{align*}
    d\chi\wedge d\phi = 0.
\end{align*}
That is $d\chi=0$ mod $S^1$, namely $\chi$ is closed mod $S^1$. We also have that $\chi$ is Dirichlet mod $S^1$. Indeed, because $\beta$ and $\Omega$ are Dirichlet but $\boundary^*d\phi\neq 0$, we have that
\begin{equation*}
0=\boundary^*\beta=\cancel{\boundary^*\Omega} + \boundary^*\chi\wedge \boundary^* d\phi \iff \boundary^*\chi=0.
\end{equation*}
Now, letting $t \in S^1$, $x \in \Sigma$ and $v \in T_x \partial \Sigma$, we consider the vector $V = d(\btob_t)_x v$ where $\btob_t : \partial \Sigma \to \partial (S^1\times\Sigma)$ is inclusion inherited from $\inclusion_t : \Sigma \to S^1\times\Sigma$ and the inclusions $\boundary : \partial (S^1\times\Sigma) \subset S^1\times\Sigma$ and $\chipboundary : \partial \Sigma \subset \Sigma$. Then, we find that
\begin{equation*}
0 = \intprod_V(\boundary^*\chi\wedge \boundary^* d\phi)|_x = \boundary^*\chi|_x(V) \boundary^* d\phi|_x - \cancel{\boundary^* d\phi|_x(V)\boundary^*\chi|_x} 
\end{equation*}
so that $\chi|_x(v) = 0$. Hence, the one-form $\inclusion_t^*\chi$ is Dirichlet for each $t \in S^1$, as claimed. We therefore have that
\begin{equation*}
\Psi^*\beta = \omega + \nu \wedge d\phi
\end{equation*}
whereby $\nu = \Psi^*\chi$ is both Dirichlet- and closed- mod $S^1$. Considering the smooth family $\{\nu_t = \inclusion_t^*\nu\}_{t \in S^1}$, we now summarise our immediate findings for the cosymplectic manifold $(S^1\times\Sigma,\beta,d\phi)$ from \cref{thm:moser}.
\end{remark}

\begin{corollary}
The cosymplectic manifold $(S^1\times\Sigma,\beta,d\phi)$ is represented by the non-autonomous locally-Hamiltonian system $(\Sigma,\omega,\{\nu_t\}_{t \in S^1})$.
\end{corollary}

By virtue of \cref{prop:symplectic-is-hamiltonian}, more can be said in the case where $\Sigma$ is planar, that is, a compact regular domain in $\mathbb{R}^2$.

\begin{corollary}\label{cor:planar-is-ham}
Assume that $\Sigma$ is planar. Then, the cosymplectic manifold $(S^1\times\Sigma,\beta,d\phi)$ is represented by a non-autonomous Hamiltonian system.
\end{corollary}

By collecting and expanding the meaning of these results, we summarise our findings in the proof of the main result, \cref{thm:moserForPlanarSigma}.
\begin{proof}[Proof of \cref{thm:moserForPlanarSigma}.]
Recalling the assumptions, let $(S^1\times\Sigma,\beta,d\phi)$ be a cosymplectic manifold where $\Sigma$ is planar, $\phi : S^1\times\Sigma \to S^1$ is the trivial projection, and $\beta$ satisfies the Dirichlet boundary condition, $\boundary^*\beta = 0$. Fix a reference angle $o\in S^1$. Implicit in \cref{cor:planar-is-ham} from the planarity of $\Sigma$ is the result that the closed one-form $\nu$ on $S^1\times\Sigma$ in \cref{rem:chi-nu} is actually exact. That is, $\nu=-dH$ for some function $H$ in $S^1\times\Sigma$. Hence, $(\Sigma,\inclusion_o^*\beta,\{\inclusion_t^*H\}_{t\in S^1})$ is a non-autonomous Hamiltonian system. The diffeomorphism of \cref{thm:moser}, which preserves the projection $\Psi^*\phi=\phi$, allows us to write
\begin{align*}
    \Psi^*\beta = \omega - dH\wedge d\phi\iff \beta = \Omega - d\mathcal{H}\wedge d\phi
\end{align*}
where $\Psi^*\Omega=\omega=\pi^*\inclusion_o\beta$ and $\Psi^*\mathcal{H}=H$.
\end{proof}
%%% Local Variables:
%%% mode: latex
%%% TeX-master: "main"
%%% End:

%%%%%%%%%%%%

\section{Conclusion}
\label{sec:conclusion}

In summary, this paper carefully investigated the relation between the field-line dynamics of divergence-free tangential vector fields on a fibre bundle over the circle and non-autonomous Hamiltonian systems. 

First, various flavours of non-autonomous Hamiltonian systems (periodic in time) were presented. Essentially, the evolution vector field $E_H$ of a non-autonomous Hamiltonian system on a symplectic manifold $(\Sigma,\omega_o)$ induced by the family of functions $\{H_t\}_{t\in S^1}$ is the Reeb field $E_H=\partial_t + X_H=R_o$ of the cosymplectic manifold $(S^1\times \Sigma,\pi^*\omega_o-dH\wedge d\phi, d\phi)$, see \cref{prop:non-autonomous-hamiltonian-systems} and \cref{rem:cosymplectic-beta}. Of course, the non-autonomous Hamiltonian structure and the corresponding dynamics on $S^1\times\Sigma$ should be insensitive to diffeomorphisms; if a cosymplectic manifold $(M,\beta,\eta)$ can be converted to the one of the above type via a diffeomorphism, then we say that it admits a non-autonomous Hamiltonian representation, see \cref{def:cosymplectic-hamiltonian-system}. Paraphrasing \cref{rem:beta_reeb}, the closed two-form then necessarily decomposes as $\beta=\Omega-d\mathcal{H}\wedge \eta$ and the Reeb field as $R=\timelike + \hamlike$ where the time-like and Hamiltonian components are characterised by
\begin{align*}
    \timelike&\in\ker\Omega,& \eta(\timelike)&=1,\\
    \hamlike&\in\ker\eta, & \intprod_{\hamlike}\Omega &= -d\mathcal{H} + d\mathcal{H}(T)\eta.
\end{align*}
In this formalism, bundle automorphisms of $S^1\times \Sigma$ induce, in some sense, non-canonical transformations; the family of functions $\mathcal{H}_t=\inclusion^*_t\mathcal{H}$ and the family of symplectic two-forms $\omega_t=\inclusion^*_t\beta$ produce dynamics on $\Sigma$ given by the flow of the vector field $\mathcal{X}_t$ satisfying
\begin{align*}
    \intprod_{\mathcal{X}_t}\omega_t = -d\mathcal{H}_t,
\end{align*}
which can be viewed as a generalisation of the dynamics of a non-autonomous Hamiltonian system, see \cref{sec:bundle-isomorphism}. 

We then defined the notion of non-autonomous locally-Hamiltonian system, for which the dynamics is induced by the family of Hamiltonian vector fields $\{X_t\}_{t\in S^1}$ satisfying
\begin{align*}
    \intprod_{X_t}\omega_o = \nu_t,
\end{align*}
where $\{\nu_t\}_{t\in S^1}$ is a family of closed Dirichlet one-forms on $\Sigma$, see \cref{def:cosymplectic-symplectic-system}. By extension, a cosymplectic manifold $(M,\beta,\eta)$ admits a non-autonomous locally-Hamiltonian representation if there exists a diffeomorphism back to a cosymplectic manifold of the form $(S^1\times\Sigma,\pi^*\omega_o+\nu\wedge d\phi,d\phi)$, where $\nu$ is a one-form on $S^1\times\Sigma$ such that $i^*_t\nu=\nu_t$, and the diffeomorphism is required to respect the cosymplectic structure. It was highlighted that $\nu$ could be chosen closed if and only if the cohomology classes of $\{\nu_t\}_{t\in S^1}$ agree. 
If $\Sigma$ is a regular compact domain of $\R^{2n}$, each closed one-form $\nu_t$ is actually exact, as a consequence of the Dirichlet boundary condition and injectivity of the boundary inclusion map on the first cohomology. Then, a (global) function $H$ on $S^1\times\Sigma$ can be found so that $\nu=-dH$ and the notions of non-autonomous locally- and globally- Hamiltonian coincide, see \cref{prop:symplectic-is-hamiltonian}.

In the second part of the paper, we attempted to reverse the argument by considering a divergence-free tangential vector field $B$ on a compact connected Riemannian three-manifold $M$. The manifold must at least be a fibre bundle over the circle. This fibre-bundle structure comes about if the field is transverse to a global Poincaré section $\Sigma$. In this case, the projection (angle function) $\phi:M\to S^1$ can be arranged so that $d\phi(B)>0$, see \cref{lem:equivalence-anglefunction}. Viewing $(M,\intprod_B\mu,d\phi)$ as a cosymplectic manifold, we asked whether it can be represented by some flavour of non-autonomous Hamiltonian system. We focused on the situation in which $M$ is a trivial bundle over the circle, where a ``time" direction can be defined. By \cref{cor:disk-annulus-trivial}, this is the case if the sections $\Sigma$ are diffeomorphic to the disk or the annulus. An adaptation of Moser's trick allowed us to conclude that $(S^1\times\Sigma,\intprod_B\mu, d\phi)$ is always represented by a non-autonomous locally-Hamiltonian system, see \cref{thm:moser}. In the case $\Sigma$ is planar, \cref{prop:symplectic-is-hamiltonian} allowed for the full identification.

On non-trivial fibre bundles, there is no natural choice of ``time" direction; there are as many candidates as there are conjugacy classes of the mapping class group $\mathcal{M}(\Sigma)$, see \cref{sec:MappingToriMCG}. The latter need not even be finite, see \cref{cor:InfiniteOrderTwists}. Future work will investigate the idea of recording the monodromy map induced by the vector field $B$, so that the time-like component arises as an instance of a purely topological feature of the manifold akin to a geometric (Berry) phase. 

The work has potential for an extension to the equivalence between volume-preserving vector fields on a $2n+1$ manifold and non-autonomous volume-preserving vector fields on a $2n$-manifold, with possible implications for Nambu mechanics. 

The work has immediate implications for the design and optimisation of magnetic confinement fusion devices. For instance, it must be recognised that the dynamics of divergence-free vector fields on toroidal domains is considerably richer than that of periodic non-autonomous Hamiltonian systems, and that the tools of Hamiltonian mechanics have limited scope in the comprehensive characterisation of magnetic configurations. In different terms, non-autonomous Hamiltonian systems only represent a subset of all possible magnetic field-line dynamics, and there is no reason to believe that the former offers better confinement properties than the latter. It is clear that the parameter space of three-dimensional stellarator magnetic fields is significantly greater than that of the axisymmetric tokamak concept. Exploration and optimisation of stellarator designs is guided by the intuition offered by the Hamiltonian interpretation, for example in the prediction/control of the onset of magnetic islands and stochastic fields, or of the persistence of invariant surfaces under perturbations. Perhaps though, configurations without classical symmetry exhibit outstanding confinement properties and may be less sensitive to small adjustments. In probing the design space for this eventuality, the theoretical framework and numerical tools needs to be extended to general classes of magnetic fields and less elementary domain topologies.

% \begin{itemize}
%     \item It would be nice to do this for symplectic manifolds of $2n$, $n>1$ dimension. The difficulty is in using the Moser trick. There are additional challenges involving when a volume form on a $2n$ manifold can be written as the product of symplectic forms. In this setting is essential to consider the subset of volume-preserving vector fields on $2n+1$ manifold that are presymplectic. 

%     \item For non-trivial fibre bundles there is no natural choice of time. However, there is an equivalent of Thm 4.9 which says that you can find a `time-like' vector field with whatever monodromy map $f$ you desire, provided that it is conjugate to some element in the mapping class group of the monodromy map induced by $B$. This forces the dynamics of $B$ to have a component $\timelike$ which comes purely from the topology of the manifold. This is a kind of geometric phase (or Berry phase).
%     \item It is probably possible to come up with an extension of a holed disk that has a Hamiltonian. 
%     \item HARD but interesting. Suppose that you have a Poincar\'e section diffeomorphic to a disk. We show that any area preserving diffeo can come from a magnetic field. What area preserving diffeo can come from magnetic fields satisfying force-balance constraints. For example, if it is MHS with $dp\neq 0$ then the map must be integrable. If it is Beltrami with constant proportionality factor or better yet vacuum, I am not sure.
% \end{itemize}

% %%%%%%%%%%%%%%%%%%%%%%%%%%%%%%

\section*{Acknowledgements}
The authors are indebted to H. Dullin, J.D. Meiss, P.J. Morrison, S.R. Hudson, A.H. Boozer, A. Bhattacharjee, P. Helander and many other members of the Simons Collaboration on Hidden Symmetries and Fusion Energy for insightful discussions. The latter collaboration is acknowledged for nurturing strong synergies and providing multiple opportunities for the authors to meet. A special thought goes to the late R.L. Dewar, whom we believe would have enjoyed the exposition.

%%%%%%%%%%%%%%%%%%%%%%%%%%%%%%%%%%%%

\bibliographystyle{apsrev4-2}
\bibliography{biblio}

\appendix

\section{Dirichlet closed one-forms with non-zero cohomology classes}
\label{eg:varying-cohomology-classes}
\begin{example}
Given a closed $n$-manifold $\Sigma_0$ with $H^1_{\text{dR}}(\Sigma_0) \neq \{ 0 \}$, we remove a ball from $\Sigma_0$ to construct a manifold $\Sigma_1$ with boundary which admits a closed one-form $\nu$ on $\Sigma_1$ with non-zero cohomology class and satisfying $\boundary^*\nu = 0$.

On the closed manifold $\Sigma_0$, there exists a smooth closed curve $\gamma : [0,1] \to \Sigma_0$ and a closed one-form $\eta$ such that $\int_\gamma \eta \neq 0$. Now, taking a point $x \in \Sigma_0 \setminus \gamma([0,1])$, consider a chart $(U,\varphi)$ about $x$ with $U \subset \Sigma_0 \setminus \gamma([0,1])$, $\varphi(x) = 0$, and $\varphi(U) = B_{\delta}$ for some $\delta > 0$ (where $B_r$ denotes an open ball of radius $r > 0$ in $\mathbb{R}^n$ centred at $0$). By the Poincar\'e Lemma, we have that $\inclusion^*\eta = dh$ for some $h \in C^{\infty}(U)$ where $\inclusion : U \subset \Sigma_0$ denotes the inclusion. Now, let $\rho$ be a smooth function in $B_{\delta}$ which is $1$ on $B_{\delta/3}$ and zero outside $B_{\delta/2}$. Now define $f : \Sigma_0 \to \mathbb{R}$ by
\begin{equation*}
f(x) = 
\begin{cases}
    \rho(\varphi(x))h(x), & x\in U\\
    0, & x\in \Sigma_0\setminus U
\end{cases}
\end{equation*}
Then, $f \in C^{\infty}(\Sigma_0)$ and $\eta|_V = df|_V$ where $V = \varphi^{-1}(B_{\delta/3})$. Let $\nu_0 = \eta - df$. By construction, $\nu_0$ is closed on $\Sigma_0$ and $\nu_0|_V= 0$ and hence $\nu_0|_{\overline{V}} = 0$.

Consider the manifold $\Sigma_1 = \Sigma_0 \setminus V$ with boundary $\partial \Sigma_1 = \overline{V} \setminus V$. Letting $\inclusion_{10} : \Sigma_1 \subset \Sigma_0$ denote the inclusion, we obtain a closed one-form $\nu = \inclusion_{10}^*\nu_0$ on $\Sigma_1$. We have that
\begin{equation*}
\int_\gamma \nu_0 = \int_\gamma \eta \neq 0    
\end{equation*}
and because $\gamma([0,1]) \subset \Sigma_1$, $\nu$ has non-trivial cohomology class. Moreover, because $\nu_0|_{\partial \Sigma_1} = 0$, it follows that $\nu$ satisfies the Dirichlet boundary condition on $\partial \Sigma_1$.

In particular, this construction applies to closed manifolds which are symplectic, and therefore one trivially obtains symplectic manifodls with boundary admitting Dirichlet closed one-forms with non-trivial cohomology class.
\end{example}

\section{Proof of \cref{lem:exactfamily}}
\label{sec:exact-is-exact}
To prove \cref{lem:exactfamily}, we first introduce a technical lemma about smoothness.
\begin{lemma}\label{lem:exactsmoothfamily}
Consider a star-shaped open subset $U$ with respect to the origin $0$ of either $\mathbb{R}^n$ or the upper half-plane $\mathbb{H}^n$ and a smooth family $\{\nu_t\}_{t \in S^1}$ of closed one-forms on $U$. Let $h : S^1 \to \mathbb{R}$ be a smooth function. Then, there exists a unique smooth family $H_t$ of functions on $U$ such that
\begin{equation*}
-dH_t = \nu_t, \qquad H_t(0) = h(t).
\end{equation*}
\end{lemma}

\begin{proof}
Indeed, for each $t \in S^1$, by Poincar\'e Lemma, we consider the function $H_t : U \to \mathbb{R}$ given by
\begin{equation}\label{smoothformla}
H_t(x) = h(t) - \int_{\gamma_x} \nu_t
\end{equation}
where $\gamma_x : [0,1] \to U$ is given by $\gamma_x(t) = t x$ for $x \in U$. We note that $H_t$ is the unique function such that
\begin{equation*}
-dH_t = \nu_t, \qquad H_t(0) = h(t).    
\end{equation*}
However, by \cref{smoothformla}, the function $H : S^1 \times U \to \mathbb{R}$ given by
\begin{equation*}
H(t,x) = H_t(x)
\end{equation*}
is smooth.
\end{proof}

We are now ready to prove \cref{lem:exactfamily}.

\begin{proof}[Proof of \cref{lem:exactfamily}]
First, fix a point $x_0 \in \Sigma$. Then, for each $t \in S^1$, by exactness of $\nu_t$, there exists a unique $H_t \in C^\infty(\Sigma)$ such that
\begin{equation*}
dH_t = \nu_t, \qquad H_t(x_0) = 0.  
\end{equation*}
Consider the function $H : S^1 \times \Sigma \to \mathbb{R}$ such that
\begin{equation*}
H(t,x) = H_t(x).
\end{equation*}
Now, let $x \in \Sigma$ and let $\gamma : [0,1] \to \Sigma$ be a smooth path from $x_0$ to $x$. Then, we have that
\begin{equation*}
H(t,x) = \int_\gamma \nu_t
\end{equation*}
and therefore the function $h_x : S^1 \to \mathbb{R}$ given by
\begin{equation*}
h_x(t) = H_t(x) = H(t,x)
\end{equation*}
is smooth because the map $\nu : S^1 \times \Sigma \to T^*\Sigma$ given by $\nu(t,x) = \nu_t|_x$ is smooth. Thus, we may use a chart about $x$ to apply \cref{lem:exactsmoothfamily} to the setting of our proposition to establish smoothness of $H$. That is, $\{H_t\}_{t \in S^1}$ is a smooth family.
\end{proof}

\section{Necessary planarity of a factor in an embedded product}
\label{sec:planar-if-embedded}
\begin{proposition}
Let $\Sigma$ be a compact two-manifold with boundary. If $S^1\times\Sigma$ can be embedded in $\R^3$, then $\Sigma$ can be embedded in $\R^2$ (planar).
\end{proposition}

\begin{proof}
First note that it suffices to consider the case where $\Sigma$ is connected. Indeed, assuming the result holds for this case, if $\Sigma$ is a compact two-manifold with boundary such that $S^1\times\Sigma$ can be embedded in $\R^3$, then, by compactness, there are only finitely many connected components of $\Sigma$. Let $\Sigma_1,...,\Sigma_k$ denote the components. Note that each connected component $\Sigma_i$ of $\Sigma$ is such that $S^1 \times \Sigma_i$ can be embedded in $\mathbb{R}^3$. Therefore $\Sigma_i$ can be embedded in $\mathbb{R}^2$. Since the $\Sigma_i$ are compact and mutually disjoint, it follows that $\Sigma = \cup_{i=1}^k \Sigma_i$ can also be embedded in $\R^2$.

From here, the proof is divided into two steps. Step 1 is showing that the inclusion map $\boundary:\partial \Sigma\to \Sigma$ is injective on first cohomology. Step 2 is to use the classification of surfaces (together with the necessary orientability of $\Sigma$) to conclude that $\Sigma$ must be planar.

\emph{Step 1.} The image $M\subset \R^3$ of an embedding is a compact regular domain. By \cref{prop:boundarycohomology}, the boundary inclusion $\partial M\to M$ is injective on first cohomology. Since $\partial( S^1\times\Sigma) = S^1\times\partial \Sigma$, we also have that $l: S^1\times \partial \Sigma\to S^1\times \Sigma$ is injective on first cohomology. 

We now show that $\boundary:\partial \Sigma\to \Sigma $ is injective on first cohomology. First, denote $\pi_\partial:S^1\times \partial \Sigma\to \partial \Sigma$. Note that $\boundary\circ \pi_\partial = \pi\circ l$, i.e. $\pi_\partial^*\boundary^* = l^*\pi^*$ as pullback maps on one-forms. Also recall that $\pi\circ \inclusion_o=\id $, i.e. $\inclusion_o^*\pi^*=\id$ for any $o\in S^1$.

Let $\nu\in\Omega^1(\Sigma)$ be a closed one-form such that $\boundary^*\nu$ is exact. Consider the closed one-form $\tau=\pi^*\nu \in\Omega^1(S^1\times \Sigma)$. The one-form
\begin{align*}
l^*\tau = l^*\pi^*\nu
= \pi_\partial^*\boundary^* \nu
\end{align*}
is thus exact. By injectivity of $l^*:H^1_{dR}(S^1\times \Sigma)\to H^1_{dR}(S^1\times \partial \Sigma)$, $\pi^*\nu=\tau$ is exact. Hence, $\nu=\inclusion_o^*\pi^*\nu$ is exact too. This proves injectivity of $\boundary : \partial \Sigma \subset \Sigma$ on cohomology.

\emph{Step 2.} We first note that $\Sigma$ is orientable and has non-empty boundary because $M \subset \R^3$ is orientable and has non-empty boundary. From Step 1, we know that $\boundary : \partial \Sigma \subset \Sigma$ is injective on cohomology. Now, each connected component of $\partial \Sigma$, as a closed connected one-manifold, is diffeomorphic to $S^1$. We form a two-manifold $S$ by patching each hole, detected by the connected components of $\partial \Sigma$, with a disk. More precisely, if there are $m$-connected components of $\partial \Sigma$, then we take a diffeomorphism $f : \partial \Sigma \to \sqcup_{i=1}^m \partial D_i$ where each $D_i$ is a disk, and form the adjunction $S = \Sigma \cup_{f} \sqcup_{i=1}^m D_i$, see \cite[Theorem 9.29]{Lee_2012}.

It follows that $S$ is a compact orientable two-manifold and is therefore a surface of genus-$g$. Because the components of $\partial \Sigma$ that are embedded in $S$ are contractible in $S$ and $\boundary$ is injective on first cohomology (surjective on first homology), it follows that the first homology of $S$ is trivial. Therefore, $g = 0$ and $S$ is a sphere. We have by construction that $\Sigma$ is a proper subset of $S$. That is, $\Sigma \subset S \setminus \{x\}$ for some $x \in S$. Thus, by stereographic projection, $\Sigma$ can be embedded in $\mathbb{R}^2$.
\end{proof}
%%%%%%%%%%%%%%%%%%%%%%%%%%%%%%%%%%
%% Here is the EDITED version %%
%%%%%%%%%%%%%%%%%%%%%%%%%%%%%%%%%%

\section{Proof of \cref{lem:equivalence-anglefunction}}
\label{sec:proof-equivalence-anglefunction}

In this section, a proof of \cref{lem:equivalence-anglefunction} is given. Different formulations of the lemma with varying hypothesis have been given in several previous works, with the first instance of the result stated by Birkhoff \cite[Ch.~V.10]{birkhoffDynamicalSystems1927}. See, for example, the work of \cite{schwartzmanGlobalCrossSections1962,Tischler_1970,simicCrosssectionsFlowsIntrinsically2023}. Here we state and prove it for compact manifolds with boundary. 

Specifically, it shown that, if $M$ is a compact orientable manifold with boundary $\partial M$ and $X$ is a vector field on $M$ tangent to $\partial M$ that admits a global Poincar\'e section $\Sigma$, then there exists an angle valued function $\phi:M\to S^1$ such that $d\phi(X) > 0$ and $\Sigma$ is a level-set of $\phi$. By a global Poincar\'e section of $X$ on a manifold with boundary $M$, we mean a compact connected embedded submanifold $\Sigma$ of $M$ such that $\partial \Sigma \subset \partial M$ and every orbit of $X$ intersects $\Sigma$.

The proof proceeds as follows. First, a Flowout Theorem for manifolds with boundary is proved in \cref{lem:flowout-thm}. The proof relies upon \cref{lem:extending-sections} which shows that $\Sigma$ can be extended to a submanifold $\tilde{\Sigma}$ without boundary of the boundaryless double $D(M)$, and $X$ can be extended to a $\tilde{X}$ on $D(M)$ such that $\tilde{X}$ is everywhere transverse to $\tilde{\Sigma}$. Using this extension, it is shown how the Flowout Theorem for manifolds without boundary \cite[Theorem 9.20]{Lee_2012} implies the Flowout Theorem for manifolds with boundary. The Flowout Theorem is then used to prove \cref{lem:first-return-time} which states that the classical notion of first return time $\tau(x) := \min\{ t> 0\, :\, \psi_t^X(x) \in \Sigma \}$ is well-defined and smooth for manifolds with boundary. In turn, \cref{lem:first-return-time} is used to show that the Poincar\'e map $f(x) := \psi^X_{\tau(x)}(x)$ is a diffeomorphism of $\Sigma$ in \cref{prop:return-is-diffeo}. Finally, the smoothness of $\tau(x), f(x)$ together with the concept of a mapping torus and special flow (see \cref{sec:MappingToriMCG}) are used to prove the existence of the desired circle-valued function $\phi:M\to S^1$.

\begin{lemma}\label{lem:extending-sections}
Let $M$ be a manifold (without boundary) and let $\Sigma$ be a compact embedded manifold of codimension $1$ that has boundary in $M$. Let $X$ be a vector field on $M$ which is transverse to $\Sigma$. Then, there exists an embedded manifold $\tilde{\Sigma}$ without boundary of $M$ and of codimension $1$ such that $\Sigma \subset \tilde{\Sigma}$ and $X$ is transverse to $\tilde{\Sigma}$.
\end{lemma}

\begin{proof}
Let the dimension of $M$ be $n+1$ so that $\Sigma$ has dimension $n$. As shown in \cite[Theorem 5.51]{Lee_2012}, for each $x \in \Sigma$, there exists a so called \emph{slice chart} $(U,\varphi)$ with $x \in U$ such that, if $x \in \Int \Sigma$, 
\begin{equation*}
    \Sigma \cap U = \{x \in U : \varphi^{n}(x) = 0\},   
\end{equation*}
and if $x \in \partial \Sigma$,
\begin{equation*}
\Sigma \cap U = \{x \in U : \varphi^{n}(x) = 0, \varphi^{n+1}(x) \geq 0\}.   
\end{equation*}
For any function $f : U \to \R$ where $U$ is an open subset of $M$, denote by $d_* f$ the local 1-form $U \to T^*M$ such that $(d_* f)|_x = (di|_x^{-1})^*df|_x$ where $i : U \subset M$ is the inclusion and $di|_x : T_x U \to T_x M$ is the associated tangent map (which is an isomorphism for slice charts).

We will say a slice chart $(U,\varphi)$ at $x \in \Sigma$ is \emph{proper} if $\overline{U}$ is a compact subset of $M$ and $d_*{\varphi}^n (X|_{U}) > 0$. We now show that a proper slice chart exists for all $x\in \Sigma$.
Indeed, let $(U,\varphi)$ be any slice chart at $x$. Consider the local 1-form $d_*\varphi^n : U \to T^*M$. By transversality of $X$, we have that $d_*\varphi^n(X|_U)$ has no zeros. In particular, for any open ball $B \subset \R^{n+1}$ such that $\bar{B} \subset \varphi(U)$ and $\varphi(x)\in B$ it holds that either $d\varphi^n(X|_{\varphi^{-1}(B)}) > 0$ or $d_*\varphi^n(X|_{\varphi^{-1}(B)}) < 0$. That is, $d_*\sigma \varphi^n(X|_V) > 0$ for some $\sigma \in \{\pm 1\}$. Then, consider the chart $(U'=\varphi^{-1}(B),\varphi')$ where ${\varphi'}^i = \varphi^i|_{U'}$ for $i \in \{1,...,n+1\} \setminus \{n\}$ and ${\varphi'}^n = \sigma \varphi^n|_{U'}$. We have that $\overline{U'}$ is a compact subset of $M$ and that the local 1-form $d_*{\varphi'}^n : U' \to T^*M$ is such that $d_*{\varphi'}^n (X|_{U'}) > 0$, that is, $(U',\varphi')$ is a proper slice chart at $x$.

The compactness of $\Sigma$ guarantees that there exists $N \in \mathbb{N}$ and proper slice-charts $\{(U_i,\varphi_i)\}_{i=1}^N$ such that $\{U_1,...,U_N\}$ forms an open cover of $\Sigma$. Now, $\{U_1,...,U_N\}$ trivially forms an open cover for the open submanifold $S = \cup_{i=1}^N U_i \subset M$ and so there exists a partition of unity $\{\rho_i : S \to \R\}_{i = 1}^N$ subordinate to this open cover. Now, for $i \in \{1,...,N\}$, we have that $\supp \rho_i \subset U_i$. Therefore, the function $F_i : M \to \R$ given by
\begin{equation*}
F_i(y) = 
\begin{cases}
\rho_i(y)\varphi_i^n(y) \text{ if } y \in U_i\\
0 \text{ otherwise},
\end{cases}
\end{equation*}
is smooth and compactly supported.

Now let $F = \sum_i F_i$ and observe that, if $y \in \Sigma$ then $F(y) = 0$. Moreover, because $dF_j|_y = 0$ for all $j \in \{1,...,N\}$ such that $y \notin \supp \rho_j$, we can write
\[ dF|_y (X|_y) = \sum_{\{i : y\in\supp\rho_i\}} \varphi^n_i(y) (d_* \rho_i)|_y(X|_y) + \rho_i(y) d_* \varphi_i^n(X|_y) = \sum_{\{i : y\in\supp\rho_i\}} \rho_i(y) d_* \varphi_i^n(X|_y) > 0,\]
because for all $i \in \{1,...,N\}$, $\varphi_i^n(y) = 0$, $d_* \varphi_i^n(X|_y) > 0$ by definition of a proper slice chart, and $\rho_i(y) > 0$ if $y\in \supp \rho_i$. 

Now, let $V = \{x \in M : dF(X)|_x > 0\}$. Then, $V$ is an open subset of $M$. The function $f = F|_V$ has no critical points. Therefore $f^{-1}(0)$ is a regular level set of $f$. However, we have from the fact that $dF|_\Sigma(X|_\Sigma) > 0$ that $\Sigma \subset V$. Therefore, $\Sigma \subset f^{-1}(0)$. The result now follows from the Preimage Theorem by setting $\tilde{\Sigma} = f^{-1}(0)$.
\end{proof}

Using \cref{lem:extending-sections}, we will now prove a flowout theorem for manifolds with boundary. 

\begin{lemma}\label{lem:flowout-thm}
Let $M$ be a compact manifold with boundary. Let $\Sigma$ be a compact codimension 1 submanifold with boundary of $M$ with $\partial \Sigma \subset \partial M$. Let $X$ be a vector field on $M$ that is tangent to $\partial M$, transverse to $\Sigma$, and let the global flow of $X$ be denoted by $\psi^X$. Then, there exists $\delta > 0$ such that, for any $t \in \R$, the flow $\psi^X$ restricted to $(t-\delta,t+\delta) \times \Sigma$ is a diffeomorphism onto its open image.
\end{lemma}

In particular, this makes $\Sigma$ transverse to the flow $\psi^X : \R \times M \to M$ in the topological sense of \cite{schwartzmanGlobalCrossSections1962}.

\begin{proof}
First, consider the boundaryless double $D(M)$ of $M$ (for an introduction to the boundaryless double, see \cite[Example 9.32]{Lee_2012}. We have that $D(M)$ is a manifold without boundary and that there is an embedding $i_1 : M \to D(M)$ whereby $i_1(M)$ is a compact regular domain of $D(M)$. By the Extension Lemma for Vector Fields (see \cite[Lemma 8.6]{Lee_2012}), there exists a vector field $\tilde{X}$ extending $X$ to $D(M)$, in the sense that $\tilde{X} \circ i_1 = di_1 \circ X$. 

Now, $i_1(\Sigma)$ is a compact embedded submanifold with boundary in $M$ of codimension 1 and $\tilde{X}$ is a vector field on $D(M)$ that is transverse to $i_1(\Sigma)$. Hence, by Lemma \ref{lem:extending-sections}, there exists an embedded submanifold $\tilde{\Sigma}$ of dimension $n$ in $D(M)$ whereby $i_1(\Sigma) \subset \tilde{\Sigma}$ and $\tilde{X}$ is transverse to $\tilde{\Sigma}$. Then, by the (classical) Flowout Theorem \cite[Theorem 9.20]{Lee_2012}, there exists a smooth positive function $\tilde{\delta} : \tilde{\Sigma} \to \R$ such that the restriction of the global flow $\psi^{\tilde{X}} : \R \times D(M) \to D(M)$ to 
\begin{equation*}
    \tilde{U} = \{(t,x) \in \R \times D(M) : |t| < \tilde{\delta}(x)\}    
\end{equation*}
is a diffeomorphism onto its open image. Now, consider the smooth positive function $\delta = \tilde{\delta} \circ i_1$ and the open set
\begin{equation*}
U = \{(t,x) \in \R \times M : |t| < \delta(x)\}.    
\end{equation*}
Observe that the flows $\psi^X$ and $\psi^{\tilde{X}}$ are related by the fact that, for $x \in M$ and $t \in \R$,
\begin{equation*}
\psi^{\tilde{X}}(i_1(x),t) = \psi^{X}(x,t).
\end{equation*}
In particular,
\begin{equation*}
i_1(\psi^{X}(U)) = \psi^{\tilde{X}}(\tilde{U}) \cap i_1(M),  
\end{equation*}
so that $\psi^{X}(U)$ is an open subset of $M$. It follows that $\psi^{X}$ restricted to $U$ is a diffeomorphism onto its open image. Moreover, because $\Sigma$ is compact, we have for some $\delta_0 > 0$ that $(-\delta_0,\delta_0) \times M \subset U$. Thus, $\psi^X$ restricted to $(-\delta_0,\delta_0) \times M$ is a diffeomorphism onto its open image. Lastly, for any $t \in \R$, we have
\begin{equation*}
    \psi^X|_{(t-\delta_0,t+\delta_0) \times M} = \psi_t^X \circ \psi^X|_{(-\delta_0,\delta_0) \times M} \circ T_t,
\end{equation*}
where $T_{-t} : (t-\delta_0,t+\delta_0) \times M \to (-\delta_0,\delta_0) \times M$ is the map given by $(s,x) \mapsto (s-t,x)$. Therefore, $\psi^X|_{(t-\delta_0,t+\delta_0) \times M}$ is also a diffeomorphism onto its open image.
\end{proof}

Using the flowout theorem for manifolds with boundary, we can establish the classical notion of the first return time for manifolds with boundary.

\begin{lemma}\label{lem:first-return-time}
Let $M$ be a compact manifold with boundary. Suppose that $X$ is a vector field tangent to $M$ and that $X$ admits a Poincar\'e section $\Sigma$. Then, the \emph{first return time} $\tau:\Sigma\to \Sigma$ given by
\[ \tau(x) := \min\left\{t>0\,:\,\psi^X_t(x) \in \Sigma\right\} \]
exists for all $x\in \Sigma$ and is smooth.
\end{lemma}

\begin{proof}
Throughout, we will use the fact proved in Lemma \ref{lem:flowout-thm}, that there exists $\delta > 0$ such that, for every $t \in \R$, $\psi^X$ restricted to $(t-\delta,t+\delta) \times \Sigma$ is a diffeomorphism onto its open image. Fix such a $\delta > 0$.

We begin by showing that for all $x\in\Sigma$ the first-return time $\tau(x)$ is well-defined. This will follow provided that the set of positive return times
\[ T_x^+ := \left\{ t>0\,:\, \psi^X_t(x) \in \Sigma\right\} \]
for $x \in \Sigma$ is non-empty and satisfies a uniform discreteness condition (see Equation \eqref{eq:uniform-discreteness}). In fact, $T_x^+$ is non-empty for all $x\in M$. To see this, observe, as in \cite{schwartzmanGlobalCrossSections1962}, that the compactness of $M$ guarantees that the $\omega$-limit set of any point $x \in M$ is non-empty. In particular, there exists an increasing sequence $\{t_n\}_{n \in \mathbb{N}}$ of times with $t_n \to \infty$ such that $\psi^X_{t_n}(x) \to y$ for some $y \in M$. Now, we have that $\psi^X_t(y) \in \Sigma$ for some $t \in \R$ and so $\psi^X_{t_n+t}(x) \to \psi^X_t(y)$ by continuity of $\psi^X_t$. On the other hand, $U = \psi^X((-\delta,\delta) \times \Sigma)$ is a neighbourhood of $\psi^X_t(y) \in \Sigma$. Hence, there exists $N\in \mathbb{N}$ such that $t_N+t-\delta > 0$ and $\psi^X_{t_N+t}(x) \in U$. However, because $\psi^X|_{(-\delta,\delta)\times \Sigma}$ is a diffeomorphism onto its open image $U$, there exists $s \in (-\delta,\delta)$ such that $\psi^X_s(\psi^X_{t_N+t}(x)) \in \Sigma$. By choice of $N$, $s+t_N+t > 0$. That is, for all $x\in M$, the set $T_x^+$ is non-empty.

We now show for $x \in \Sigma$ that $T_x^+$ satisfies a uniform discreteness condition. In fact, this is true for the set of all return times
\begin{equation*}
T_x = \{t \in \R : \psi^X_t(x) \in \Sigma\}.
\end{equation*}
Indeed, let $x \in \Sigma$ and $t \in T_x$. Then, $\psi^X_t(x) \in \Sigma$. On the other hand, $\psi^X$ restricted to $(-\delta,\delta) \times \Sigma$ is a diffeomorphism onto its open image. In particular, because $\psi^X(\{0\} \times \Sigma) = \Sigma$ and $\psi^X|_{(-\delta,\delta) \times \Sigma}$ is injective, for any $s \in (-\delta,\delta) \setminus \{0\}$ and $y \in \Sigma$, we have that $\psi^X_s(y) \notin \Sigma$. Hence, for any $s \in (-\delta,\delta) \setminus \{0\}$, it follows that
\begin{equation*}
\psi^X_{s+t}(x) = \psi^X_s(\psi^X_t(x)) \notin \Sigma,
\end{equation*}
and therefore $s+t \notin T_x$. It follows that, for all $x \in \Sigma$ and $t \in T_x$,
\begin{equation}\label{eq:uniform-discreteness}
T_x \cap (t-\delta,t+\delta) = \{t\}.
\end{equation}

From the fact that $T_x^+$ is non-empty and the fact that Equation \eqref{eq:uniform-discreteness} holds for all $x\in \Sigma$, we can conclude that the minimum of $T_x^+$  exists and hence $\tau(x)$ is well-defined. Moreover, an immediate consequence of Equation \eqref{eq:uniform-discreteness} that $\tau(x)$ satisfies the inequality
\begin{equation}\label{eq:bounding-plus-T}
\tau(x) \geq \delta. 
\end{equation}

We will now show that $\tau$ is smooth in an indirect manner. First, it is shown that for each $x\in \Sigma$ and $t\in T_x^+$ that there exists a neighbourhood $V$ of $x$ in $\Sigma$ and a smooth function $T : V \to (t-\delta,t+\delta)$ such that $T(x) = t$ and $T(y) \in T_y^+$ for all $y \in V$. The function $T$ gives a `local' return time function for all $y\in V$. Using this local $T$, it is shown that $\tau$ satisfies a continuity condition. Finally, the minimality $\tau$ is then used to show that, for all $x \in \Sigma$, $\tau$ and the local return time function $T$ agree on a neighbourhood $W$ of $x$ in $\Sigma$. Consequently, $\tau$ is smooth.

\textbf{Claim 1.} For all $x \in \Sigma$ and $t \in T_x^+$ there exists a neighbourhood $V$ of $x$ in $\Sigma$ and a smooth function $T : V \to (t-\delta,t+\delta)$ such that $T(x) = t$ and $T(y) \in T_y^+$ for all $y \in V$. 

Consider
\begin{equation*}
U = \psi^X((-\delta,\delta) \times \Sigma),   
\end{equation*}
and observe that there is a smooth function $f:U \to \R$ such that, for all $x'\in U$, 
\begin{equation}\label{eq:projection-onto-sigma-property}
\psi^X_{f(x')}(x') \in \Sigma.     
\end{equation}
Indeed, because $\psi^X|_{(-\delta,\delta) \times \Sigma}$ is a diffeomorphism onto its open image $U$, for any $x'\in U$ there is a unique point $(t_{x'},y_{x'})\in (-\delta,\delta) \times \Sigma$ such that $\Psi^X_{t_{x'}}(y_{x'}) = x'$. Then, we can define a smooth function $f : U \to \R$ by $f(x') = -t_{x'}$. It follows that 
\begin{equation*}
\psi^X_{f(x')}(x') = \psi^X_{-t_{x'}}(\psi^X_{t_{x'}}(y_{x'})) = y_{x'} \in \Sigma.     
\end{equation*}
Now, let $x\in \Sigma$ and $t\in T_x^+$. We will show that $V = \Sigma \cap \psi^X_{-t}(U)$ is the desired neighbourhood of $x$ and  
\[T(y) = f(\psi^X_t(y)) + t\] gives the desired smooth function $T:V\to (t-\delta,t+\delta)$ satisfying $T(x) = t$ and $T(y) \in T^+_y$ for all $y\in V$. First, as $\psi^X_t(x) \in \Sigma$, then $\psi^X_t(x) \in U$ and in turn $x \in \psi^X_{-t}(U)$. Furthermore, $x \in \Sigma$ and so $x \in \Sigma \cap \psi_{-t}^X(U)$. It follows that $\Sigma \cap \psi_{-t}^X(U)$ is a neighbourhood of $x$ in $\Sigma$. Now, $\psi_{t}^X(x) \in \Sigma$ and therefore $f(\psi_{t}^X(x)) = 0$. Thus, $T(x) = t$. 
Moreover, for each $y \in \Sigma \cap \psi_{-t}^X(U)$, we have that $y \in \Sigma$ and $\psi_t^X(y) \in W$. Hence, by Equation \eqref{eq:projection-onto-sigma-property},
\begin{equation*}
\psi^X_{T(y)}(y) = \psi^X_{f(\psi_{t}^X(y))+t}(y) = \psi^X_{f(\psi_{t}^X(y))}(\psi_t^X(y)) \in \Sigma.    
\end{equation*}
Thus $T(y) \in T_y$. Finally, $T(y) \in T_y^+$ because $t \in T^+_x$ and Equation \eqref{eq:bounding-plus-T} guarantees that $t \geq \delta$, so that
\begin{equation*}
T(y) > t-\delta \geq \delta-\delta = 0.    
\end{equation*}
In conclusion, $V = \Sigma\cap \psi^X_{-t}(W)$ and $T(y) = f(\psi_t^X(y)) + t$ have the desired properties so that claim 1 holds. 

\textbf{Claim 2.} The first-return time $\tau$ satisfies the continuity condition that, for all $x \in \Sigma$ and for any sequence $\{y_n\}_{n \in \mathbb{N}}$ with $y_n \to x$, there exists a subsequence $\{y_{n_k}\}_{k \in \mathbb{N}}$ of $\{y_n\}_{n \in \mathbb{N}}$ such that $\tau(y_{n_k}) \to \tau(x)$.

Indeed, for $x\in \Sigma$, we have that $\tau(x)\in T_x^+$ and so we may take the neighbourhood $V$ of $x$ in $\Sigma$ and function $T : V \to (\tau(x)-\delta,\tau(x)+\delta)$ as given by Claim 1. Without loss of generality, let $\{y_n\}_{n \in \mathbb{N}}$ be a sequence in $V$ with $y_n \to x$. For all $y \in V$, because $T(y),\tau(y) \in T_y^+$, it follows from Equation \eqref{eq:bounding-plus-T} that
\begin{equation}\label{eq:ineq-with-tau-and-T}
\delta \leq \tau(y) \leq T(y) < \tau(x) + \delta.
\end{equation}
In particular, the sequence $\{\tau(y_n)\}_{n \in \mathbb{N}}$ in $[\delta,\tau(x) + \delta]$ has a convergent subsequence in $[\delta,\tau(x) + \delta]$ (by compactness of $[\delta,\tau(x) + \delta]$). Hence, there exists a subsequence $\{y_{n_k}\}_{k \in \mathbb{N}}$ of $\{y_{n}\}_{n \in \mathbb{N}}$ in $\Sigma$ and a $t^* \in [\delta,\tau(x) + \delta]$ such that
\begin{equation*}
\tau(y_{n_k}) \to t^*.   
\end{equation*}
Moreover, because $y_n \to x$ by assumption, we have that $y_{n_k} \to x$. Therefore, the sequence $\{(\tau(y_{n_k}),y_{n_k})\}_{k \in \mathbb{N}}$ in $\R \times \Sigma$ is such that
\begin{equation*}
(\tau(y_{n_k}),y_{n_k}) \to (t^*,x).
\end{equation*}
Observe that for $n \in \mathbb{N}$, we have $y_n \in \Sigma$ and $\tau(y_n) \in T_{y_n}$ so that
\begin{equation*}
\psi^X(\tau(y_n),y_n) \in \Sigma.    
\end{equation*}
However, the preimage $(\psi^X|_{\R \times \Sigma})^{-1}(\Sigma)$ is a closed subset of $\R \times \Sigma$ because $\Sigma$ is a closed subset of $M$ and $\psi^X|_{\R \times \Sigma} : \R \times \Sigma \to M$ is continuous. Therefore, the limit $(t^*,x)$ of $\{(\tau(y_{n_k}),y_{n_k})\}_{k \in \mathbb{N}}$ is such that
\begin{equation*}
\psi^X(t^*,x) \in \Sigma.    
\end{equation*}
Hence, $t^* \in T_x$. Moreover, $t^* \in [\delta,\tau(x) + \delta]$ and so $t^* \in T_x^+$. Therefore, $\tau(x) \leq t^*$. On the other hand, because $T$ is continuous and $y_{n_k} \to x$, we have
\begin{equation*}
T(y_{n_k}) \to T(x) = \tau(x).    
\end{equation*}
However, for all $k \in \mathbb{N}$, from Equation \eqref{eq:ineq-with-tau-and-T}, we have
\begin{equation*}
\tau(y_{n_k}) \leq T(y_{n_k}).
\end{equation*}
Therefore $t^* \leq \tau(x)$. Hence, $t^*=\tau(x)$ and
\begin{equation*}
\tau(y_{n_k}) \to \tau(x).
\end{equation*}
This completes the proof of Claim 2.

Finally, Claim 1 and Claim 2 guarantee the smoothness of $\tau$ as follows. Let $x \in \Sigma$ and take the neighbourhood $V$ of $x$ in $\Sigma$ and smooth function $T : V \to (\tau(x)-\delta,\tau(x)+\delta)$ whereby $T(x) = \tau(x)$ and $T(y) \in T_y^+$ for $y \in V$. Now, suppose that for any neighbourhood $W$ of $x$ in $V$, there exists $y \in W$ such that $T(y) \neq \tau(y)$. Then, there exists a sequence $\{y_n\}_{n \in \mathbb{N}}$ in $\Sigma$ with $y_n \to x$ such that $T(y_n) \neq \tau(y_n)$ for all $n \in \mathbb{N}$. For $n \in \mathbb{N}$, we have that $T(y_n) > \tau(y_n)$ and $T(y_n),\tau(y_n) \in T_{y_n}^+$ so that by Equation \eqref{eq:uniform-discreteness}, we have
\begin{equation*}
T(y_n) \geq \delta + \tau(y_n).
\end{equation*}
However, from Claim 2, there exists a subsequence $\{y_{n_k}\}_{k \in \mathbb{N}}$ of $\{y_n\}_{n \in \mathbb{N}}$ such that $\tau(y_{n_k}) \to \tau(x)$. On the other hand, because $y_{n_k} \to x$, we have that $T(y_n) \to T(x) = \tau(x)$. This implies $\tau(x) \geq \delta + \tau(x)$, which is a contradiction. Therefore, there exists a neighbourhood $W$ of $x$ in $V$ such that 
\begin{equation*}
\tau|_W = T|_W.    
\end{equation*}
It follows that $\tau$ is smooth.
\end{proof}

From Lemma \ref{lem:first-return-time}, we can introduce the notion of first return map and show that it is a diffeomorphism.

\begin{proposition}\label{prop:return-is-diffeo}
Let $M$ be a compact manifold with boundary. Suppose that $X$ is a vector field tangent to $M$ and that $X$ admits a Poincar\'e section $\Sigma$. Then, the function $f : \Sigma \to \Sigma$ given by
\begin{equation*}
f(x) = \psi^X_{\tau(x)}(x)
\end{equation*}
is a diffeomorphism.
\end{proposition}

\begin{proof}
Note that the function $F : \Sigma \to M$ given by $F(x) = \psi^X_{\tau(x)}(x)$ is smooth because $\tau$ is smooth and $\psi^X : \R \times M \to M$ is smooth. Moreover, $F(\Sigma) \subset \Sigma$ and $\Sigma$ is an embedded submanifold of $M$ with boundary. Therefore, the restricted map $f : \Sigma \to \Sigma$ is also smooth (see \cite[Theorem 5.53 and Problem 9-13]{Lee_2012}).

We show now that $f$ has a smooth inverse. Consider the vector field $Y = -X$ and note that $\Sigma$ is a Poincar\'e section for $Y$. Let $\upsilon : \Sigma \to \R$ denote the smooth first-return time as in Lemma \ref{lem:first-return-time}. As per the above paragraph, the first return map $g : \Sigma \to \Sigma$ of $Y$, given by
\begin{equation*}
g(x) = \psi^Y_{\upsilon(x)}(x),
\end{equation*}
is smooth. We now show that
\begin{equation*}
f \circ g = \operatorname{id} = g \circ f.    
\end{equation*}
It is enough to show the first equality because the second is proven identically. To this end, for $x \in \Sigma$, set
\begin{align*}
T^{X}_x &= \{t \in \R : \psi^X_t(x) \in \Sigma\},\\
T^{Y}_x &= \{t \in \R : \psi^Y_t(x) \in \Sigma\}.
\end{align*}
Now, let $x \in \Sigma$ and observe that
\begin{equation*}
\psi^Y_t(f(x)) = \psi^X_{-t}(f(x)) = \psi^X_{-t}(\psi^X_{\tau(x)}(x)) = \psi^X_{\tau(x)-t}(x).     
\end{equation*}
Therefore,
\begin{equation*}
t \in T^{Y}_{f(x)} \text{ and } t > 0 \iff \tau(x)-t \in T^{X}_{x} \text{ and } \tau(x)-t < \tau(x) \iff \tau(x)-t \in T^{X}_{x} \text{ and } \tau(x)-t \leq 0.
\end{equation*}
Hence,
$
\upsilon(f(x)) = \tau(x)
$
and
\begin{equation*}
\psi^Y_{\upsilon(f(x))}(f(x)) = \psi^X_{\tau(x)-\tau(x)}(x) = x.    
\end{equation*}
That is,
\begin{equation*}
g(f(x)) = \psi^Y_{\upsilon(f(x))}(f(x)) = x.   
\end{equation*}
\end{proof}

With \cref{lem:first-return-time} showing the first-return time $\tau$ is smooth and Proposition \ref{prop:return-is-diffeo} showing the first-return map is a diffeomorphism for manifold with boundary, we are ready to prove \cref{lem:equivalence-anglefunction}.

\begin{proof}[Proof of \cref{lem:equivalence-anglefunction}]
Let $f:\Sigma\to \Sigma$ be the first-return map, $\tau : \Sigma \to \R$ the first-return time of $X$, and $\Phi = \psi^X$ be the global flow of $X$. 
Let $\Sigma^{\tau}_f$ be the manifold with boundary defined by the quotient space
\[ \Sigma^\tau_f := \frac{\{(t,x) \in \R\times \Sigma\,|\, 0 \leq t \leq \tau(x)\}}{ (0,f(x)) \sim (\tau(x), x)}. \]
Observe that $\Sigma_f^\tau$ is diffeomorphic to $M$. In fact, a diffeomorphism is explicitly given by $\tilde{\Phi}:\Sigma_f^\tau\to M$, $[t,x] \mapsto \Phi(t,x)$. Indeed, $\tilde{\Phi}$ is well-defined as 
\[\tilde{\Phi}([\tau(x),x]) = \Phi(\tau(x),x) = f(x) = \Phi(0,f(x)) = \tilde{\Phi}([0, f(x)]).\]
Moreover, $\tilde{\Phi}$ is a local diffeomorphism by the Flowout Theorem for manifolds with boundary \cref{lem:flowout-thm}. Lastly, we note that $\tilde{\Phi}$ is injective by construction of $\tau$ and $f$ and $\tilde{\Phi}$ is surjective because $\Phi$ is surjective (as $\Sigma$ is a Poincar\'e section). Hence, as claimed, $\tilde{\Phi}$ is a diffeomorphism.

Moreover, the so-called vertical vector field $\partial_t$ on $\Sigma_f^\tau$ defined by the flow $(s,[t,x])\mapsto [t+s,x]$ is $C^\infty$-orbitally equivalent to $X$; 
\[ \tilde{\Phi}_* \partial_t = \left.\frac{d}{ds}\right|_{s=0} \tilde{\Phi}([t+s,x]) = \left.\frac{d}{ds}\right|_{s=0} \Phi(t+s,x) = X.\]
Hence, $\Sigma_f^\tau$ gives the flow of $X$ as the \emph{flow under a function} or \emph{special flow} (see \cite{katokIntroductionModernTheory1995}).

Now, consider the mapping torus $\Sigma_f$ defined here as
\[ \Sigma_f := \frac{\{(t,x) \in \R\times \Sigma\,|\, 0 \leq t \leq 2\pi\}}{ (0,f(x)) \sim (2\pi, x)}. \] 
The manifold $\Sigma_f^\tau$ and the mapping torus $\Sigma_f$ are diffeomorphic \cite[Prop.~2.2.5]{katokIntroductionModernTheory1995}. Indeed, following \cite{katokIntroductionModernTheory1995}, let $k = \min \tau$, $K = \max \tau$, $\tilde{k} = \min(\{k, 2\pi \})$ and consider any smooth function $g:[0,2\pi]\times [k,K] \to [0,K]$ such that
\begin{enumerate}
    \item $g(t,s) = t$ for $t\in [0, \tfrac{1}{4} \tilde{k}]$,
    \item $g(t,s) = t+s-2\pi$ for $t\in [2\pi - \tfrac{1}{4} \tilde{k},2\pi]$,
    \item $\frac{\partial}{\partial t} g(t,s) > 0$.
\end{enumerate}
Then, the map $\Psi:\Sigma_f \to \Sigma_f^\tau$ given by $[t,x]\mapsto [g(t,\tau(x)), x]$ is a diffeomorphism.

Furthermore, the vector field $ \tilde{\partial}_t$ on $\Sigma_f$ with flow $(s,[t,x]) \mapsto [t+s,x]$ is $C^\infty$-orbitally equivalent to the flow of $\partial_t$ on $\Sigma_f^\tau$. Indeed, we have that  
\[ \Psi_* \tilde{\partial}_t = \left.\frac{d}{ds}\right|_{s=0} \Psi([t+s,x]) = \left.\frac{d}{ds}\right|_{s=0} [g(t+s, \tau(x)), x] = u \partial_t \] 
for some function $u:\Sigma_f^\tau \to \R$ that must be strictly positive in consequence of the fact that $\frac{\partial g}{\partial t} > 0$.

Finally, define $\tilde{\phi}:\Sigma_f \to S^1$ by $\tilde{\phi}([t,x]) = [t]$ and let $\phi:M\to S^1$ be the circle-valued function given from the pushforward of $\tilde{\phi}$ by the diffeomorphism $ \tilde{\Phi}\circ\Psi:\Sigma_f \to M$.
Then, $d\tilde{\phi} (\tilde{\partial}_t) = 1 \implies d\phi(X) = (\tilde{\Phi}_*u)^{-1} >0$ and $\phi^{-1}(0) = \Sigma$. Therefore, $\phi$ is the desired circle-valued function on $M$.

\end{proof}

\section{Mapping tori, mapping class groups, monodromy, and a classification of fibre bundles}\label{sec:MappingToriMCG}

Let $\phi: M \to S^1$ and $\phi':M' \to S^1$ be two fibre bundles over a circle with respective fibres $\Sigma, \Sigma'$ that are both connected, compact, smooth manifolds (possibly with boundary). The fibre bundles are said to be \emph{bundle isomorphic} if there exists a diffeomorphism $\Psi:M\to M'$ such that $\phi = \phi'\circ \Psi$, that is, the projections $\phi,\phi'$ commute with $\Psi$. A natural question is whether there exists a classification of all fibre bundles over $S^1$ up to bundle isomorphism. In particular, it is crucial to know whether there exist \emph{non-trivial} fibre bundles, namely fibre bundles that are not isomorphic to the trivial fibre bundle $M = S^1 \times \Sigma$.

In this section, we detail a representation of fibre bundles known as the \emph{monodromy representation}. For some classes of fibre bundles, the monodromy representation is a complete invariant. This means that any two fibre bundles are isomorphic if and only if they have the same monodromy representation. In particular, the monodromy representation is a complete invariant for the class of fibre bundles considered here; that of a fibre bundle over $S^1$. A recent exposition on the monodromy representation is given in \cite{salterSurfaceBundlesTopology2020}, where surface bundles with arbitrary base manifold are treated. The more general case of flat fibre bundles is treated in \cite{moritaGeometryCharacteristicClasses2001}. In this section we consider only fibre bundles with base manifold $S^1$ and consequently provide simplifications and an alternative formulation of the monodromy representation to that given in the aforementioned references. 

The section begins by establishing the monodromy representation and its relation to mapping tori and the mapping class group of $\Sigma$. Then, the monodromy representation is used to show that there are no non-trivial fibre bundles only if $\Sigma$ is a disk or annulus. 

\subsection{The Monodromy Representation}\label{sec:MonodromyRepresentation}

Let $\phi: M \to S^1$ be a fibre bundle over $S^1$ with fibre $\Sigma$ that is a connected, compact manifold (possibly with boundary). Henceforth, bundles of this class will be referred to as \emph{$\Sigma$-bundles over $S^1$}. Let $\Diff^+(\Sigma)$ be the group of orientation preserving diffeomorphisms on $\Sigma$. Before formally defining the monodromy representation, we give an intuitive understanding. The core idea is to treat $S^1 \cong \quotient{\R}{2\pi \Z}$ so that $M$ can be `chopped' along the fibre $\phi^{-1}(0)$ to form a bundle over $[0,2\pi]$. As $[0,2\pi]$ is contractible, then this chopped bundle is trivialisable. Any trivialisation essentially `untwists' the fibre bundle. We can then associate to $M$ the diffeomorphism $f \in \Diff^+(M)$ which twists and glues $\{0\} \times \Sigma$ to $\{2\pi\} \times \Sigma$ to reform $M$. The resulting function $f$ depends on the choice of base point and trivialisation. Associating to $M$ all possible diffeomorphisms $f$ is the monodromy representation.

The monodromy representation of a $\Sigma$-bundle over $S^1$ can be constructed through the use of an Ehresmann connection. Every fibre bundle comes with the so-called \textit{vertical subbundle} of $TM$ defined by $\mathcal{V}M = \ker d\phi$. An \textit{Ehresmann connection} is a smooth choice of complement to $\mathcal{V}M$ in $TM$ at each point. The complement $\mathcal{H}M$ is called the \textit{horizontal subbundle} of $TM$. By definition $TM = \mathcal{V}M\oplus\mathcal{H}M$. It will be assumed here that $\mathcal{H}M|_{\partial M} \subset T\partial M$ so that the connection is always a \emph{complete connection}.

For $\Sigma$-bundles over $S^1$, there is a natural bijection between an Ehresmann connection and vector fields on $M$ which are everywhere transverse to the fibres of $M$ and tangent to $\partial M$. Consider $S^1 = \quotient{\R}{2\pi\Z}$ parameterised by $t$. If $\mathcal{H}M$ is given by a connection, then there is a unique vector field $\Y$, tangent to $\partial M$, defined by $d\phi(\Y) = 1$ and $\Y_x \in \mathcal{H}_xM$ for all $x\in M$. Conversely, given a vector field $\Y$ such that $d\phi(\Y) = 1$ and $\Y|_{\partial M} \in \Gamma(T\partial M)$, we can define at each point $x\in M$ the horizontal subspace by $\mathcal{H}_xM = \Span (\Y_x)$. Due to this correspondence, we will henceforth define a $\emph{connection}$ on $M$ as a choice of vector field $\Y$ such that $d\phi(\Y) = 1$ and $\Y|_{\partial M} \in \Gamma( T\partial M)$.

Every $\Sigma$-bundle over $S^1$ admits a connection. 
\begin{lemma}\label{lem:EveryBundleHasAConnection}
  Every $\Sigma$-bundle over $S^1$ admits a connection. 
\end{lemma}
\begin{proof}
  Let $\gamma:\R\to S^1$ be a covering map of $S^1$ and denote by $\gamma^*M$ the pull-back bundle of $M$, defined as a subset of $\R\times M$ by
  \[ \gamma^* M = \{ (t,x)\in \R\times M \,|\, \gamma(t) \in \phi(x) \} \]
  and with projections $\gamma^*\phi: \R\times M \to \R$ and $\pi_M: \R\times M \to M$.
  Since the base space of $\gamma^*M$ is contractible then $\gamma^*M$ is trivial. 
  Let $\tau:\R\times \Sigma \to \gamma^*M$ be a trivialisation such that $(\gamma^*\phi)\circ\tau(t,x) = t$. Denote by $\partial_t$ the vector field on $\R\times \Sigma$ with flow $\Phi_s(t,x) = (t+s,x)$. Then $\Y = (\pi_M)_*\tau_* \partial_t$ is a well-defined vector field on $M$ satisfying $d\pi(\Y) = 1$, that is, $\Y$ is a connection. 
\end{proof}

In essence, a connection describes how to move from fibre to fibre in $M$. The flow $\Phi:\R\times M \to M$ of a connection $\Y$ \emph{parallel transports} points on any fibre $\Sigma_{t_1} := \phi^{-1}(t_1)$ to another $\Sigma_{t_2} := \phi^{-1}(t_2)$ with $t_2 > t_1$ via the map $\Phi_{t_1,t_2} := \Phi(t_2-t_1,\cdot):\Sigma_{t_1} \to \Sigma_{t_2}$. In particular $\Phi_{t,t+2\pi}$ transports any fibre $\Sigma_t$ once around the bundle and back to itself. As each fibre is diffeomorphic to $\Sigma$, then local trivialisations give for each $t\in S^1$ a diffeomorphism $h_t: \Sigma \to \Sigma_t$. Through these diffeomorphisms, it is possible to consider $\overline{\Phi}_{t_1,t_2} := h_{t_2}^{-1}\circ\Phi_{t_1,t_2}\circ h_{t_1}\in \Diff^+(\Sigma)$. Note that the diffeomorphisms $h_t$ need not be smooth as a function of $t$; in fact, if the fibre bundle is non-trivial, then necessarily they will not be. 

\begin{definition}
  The map $\overline{\Phi}_{0,2\pi} \in \Diff^+(\Sigma)$ is denoted the \emph{monodromy map with respect to $\Y$} of the fibre bundle $(M,S^1,\phi,\Sigma)$.
\end{definition}

The monodromy map with respect to $\Y$ has been defined with respect to a specific choice of base point $t\in S^1$, namely $t = 0$, through the identification $S^1 = \quotient{\R}{2\pi \Z}$. We now explore the consequence of changing the base from which the monodromy map is defined. Let $t\in S^1$ be a base point. Then the monodromy map based at $t$, $\overline{\Phi}_{t,t+2\pi}$, can be computed from the monodromy map of $M$ using the semi-group property of flows,
\[\overline{\Phi}_{t,t+2\pi} = \overline{\Phi}_{0,t} \circ \overline{\Phi}_{0,2\pi} \circ \overline{\Phi}_{0,t}^{-1}. \] 
It follows that changing the base from which the monodromy map is defined amounts to conjugation by $\overline{\Phi}_{0,t}\in\Diff^+(\Sigma)$. Hence, the conjugacy class of $\overline{\Phi}_{0,2\pi}$ in $\Diff^+(\Sigma)$ is an invariant of the bundle. In fact, as we prove in what follows, it is a \emph{complete} invariant of $\Sigma$-bundles over $S^1$ with a distinguished Ehresmann connection.

\begin{definition}
  Let $\phi:M\to S^1$ and $\phi':M'\to  S^1$ be two $\Sigma$-bundles with respective connections $\Y, \Y'$. A diffeomorphism $\Psi:M\to M'$ is said to be an \emph{isomorphism of bundles with connection} provided $\Psi$ is a bundle isomorphism and $\Psi_*\Y = \Y'$.
\end{definition}

\begin{theorem}[Morita \cite{moritaGeometryCharacteristicClasses2001}]\label{thm:monodromyMapIsComplete}
  The correspondence which sends any (orientable) $\Sigma$-bundle over $S^1$ with a given connection $\Y$ to its monodromy map, induces a bijection
  \begin{align*}
    ~&\{ \text{isomorphism classes of $\Sigma$-bundles over $S^1$ with connection} \} \\
    ~&\hspace{0.5\textwidth} \cong \{ \text{conjugacy classes of $\Diff^+(\Sigma)$} \} 
  \end{align*}
\end{theorem}
One direction of the theorem is clear; if there is an isomorphism $\Psi$ of bundles with connection then the flow of the connections are $\Psi$-related by definition. In consequence, the monodromy maps must be conjugate. The converse is less clear. First, to any conjugacy class of diffeomorphisms, it must be shown that there is a $\Sigma$-bundle over $S^1$ with a connection $\Y$ whose monodromy map is in the given conjugacy class. Second, if two fibre bundles have conjugate monodromy maps, then they need to be isomorphic with connection. The first requirement can be assured through a \emph{mapping torus.}

\begin{definition} \label{def:mappingtorus}
  For any $f\in \Diff^+(\Sigma)$ the \emph{mapping torus} $\Sigma_f$ associated to $f$ is given by
  $$ \Sigma_f := \frac{\R\times \Sigma}{ (t,f(x)) \sim (t-2\pi, x)}. $$
\end{definition}

A mapping torus is naturally a $\Sigma$-bundle over $S^1$ where the projection $\phi:\Sigma_f \to S^1 = \quotient{\R}{2\pi\Z}$ is defined by $[s, x] \mapsto [s]$ with $[s,x]$ denoting the equivalence class of $(s,x)\in \R\times \Sigma$ under $(s,f(x))\sim (s-2\pi,x)$ and $[s]$  the equivalence class under $ s \sim s - 2\pi$.
Every mapping torus has a connection with monodromy map $f$. Indeed, on the space $\R\times \Sigma$ there is a flow $\Phi_s(t,x) = (t+s, x)$ with associated vector field $\partial_t$. Using the quotient map $q:\R\times \Sigma \to \Sigma_f$, we obtain a well-defined vector field $\Y = q_*\partial_t$ on $\Sigma_f$. Moreover, $d\phi(\Y) = 1$ and, if $\Sigma$ has boundary, then $\Y$ is tangent to $\partial \Sigma_f$ so that $\Y$ is a complete connection on $\Sigma_f$. Finally, observe that the monodromy map of $\Y$ is exactly $f$, as desired.

% \comment{There are two obstructions to the smoothness. The first is that $\Y$ may not be well-defined to begin with. Denote by $X = (X^t, X^x)$ any vector field on $\R\times \Sigma$. Let $dq_{t_0,x_0}$ be the differential of $q$. Take $\gamma(s) = (\gamma_t(s),\gamma_x(s)):\R\to \R\times \Sigma$ to be a curve with $\gamma(0) = (t_0, x_0)$ and $\gamma'(0) = X_{t_0,x_0}$. Let $p = (t_0,x_0)$ and observe the differential $dq$ has the following induced property from the quotient:
% \[ dq_{p}(X_p) := \left.\frac{d}{dt}\right|_{t=0} q(\gamma_t(s),\gamma_x(s)) = \left.\frac{d}{dt}\right|_{t=0} q(\gamma_t(s) + 2\pi, f\circ \gamma_x(s)) = dq_{(t_0+2\pi, f(x_0))}(X^t_{p}, f'(x_0) X^x_p). \]
% So, for $Y = q_* X$ to be well-defined we need $X_{(t_0+2\pi, f(x_0))} = (X^t_{p}, f'(x_0) X_p^x)$. For our chosen $\partial_t$ we have $X^t = 1$ and $X^x = 0$ and hence $Y$ is well-defined.\\
% The second obstruction is that $q_*$ itself is not a local diffeomorphism. But we can set up local charts so that $q_*$ is essentially the identity, and so is a local diffeomorphism.}

It can be shown that every $\Sigma$-bundle over $S^1$ with monodromy map $f$ is isomorphic to the mapping torus $\Sigma_f$.
\begin{lemma}
    Every $\Sigma$-bundle over $S^1$ with a connection $\Y$ producing a monodromy map $f$ is isomorphic to $\Sigma_f$. 
\end{lemma}
% \begin{proof}
%     Define $\Psi:\Sigma_f \to M$ by $\Psi([t,x]) = \inclusion_t \circ h_t\circ \overline{\Phi}_{0,t}(x)=\inclusion_t\circ \Phi_{0,t}\circ h_0(x)$. It is well-defined as $\overline{\Phi}_{0,2\pi} = f$ by assumption and definition of the monodromy map. $\Psi$ provides the desired isomorphism. 
% \end{proof}

\begin{proof}
    Define $\Psi: M\to \Sigma_f$ by $\Psi(p) = \left[\phi(p), h_0^{-1}\circ\Phi_{0,\phi(p)}^{-1}(p)\right]$. Observe that $\Psi$ is smooth, well-defined, and preserves fibres. Indeed, to be well-defined we require $[\phi(p) + 2\pi, f\circ h_0^{-1}\circ\Phi_{0,\phi(p)+2\pi}^{-1}(p)] = [\phi(p), h_0^{-1}\circ\Phi_{0,\phi(p)}^{-1}(p)] $. Now, $[\phi(p)] = [\phi(p) + 2\pi]$ and 
    \[h_0^{-1}\circ \Phi_{0,\phi(p)}^{-1}(p) = h_0^{-1} \circ \Phi_{0,2\pi}\circ\Phi_{0, \phi(p) + 2\pi}^{-1}(p) = \bar{\Phi}_{0,2\pi}\circ h_0^{-1}\circ \Phi^{-1}_{0,\phi(p)+2\pi}(p) = f\circ h_0^{-1}\circ\Phi^{-1}_{0,\phi(p)+2\pi}(p). \]
    It follows that $\Psi$ is well-defined.
\end{proof}

We now show the second claim required for the proof of \cref{thm:monodromyMapIsComplete}, namely, that two fibre bundles with conjugate monodromy maps are isomorphic with connection. Observe that if $\Sigma_f, \Sigma_{f'}$ are two mapping tori and $f, f'$ are conjugate, that is there exists $g \in \Diff^+(\Sigma)$ such that $g\circ f = f'\circ g$, then $\Psi([t,x]) = [t,g(x)]$ is a bundle isomorphism. This completes the proof of \cref{thm:monodromyMapIsComplete}. Hence, the monodromy map is a complete invariant of $\Sigma$-bundles over $S^1$ with a connection. 

In light of the fact that every $\Sigma$-bundle over $S^1$ admits a connection (\cref{lem:EveryBundleHasAConnection}), it is natural to view the monodromy map of a $\Sigma$-bundle as an intrinsic invariant of the bundle itself, and not associated to a specific connection. This intrinsic invariant is called the \emph{monodromy representation} of $M$. Two different choices in connection give rise to two different monodromy maps $f,f'$. It follows that, in order to make sure the monodromy representation is well-defined, we need to understand how a change in choice of connection effects the monodromy map. 

As $M$ is isomorphic to two different mapping tori $\Sigma_f, \Sigma_{f'}$, finding a condition on the monodromy representation so that it is well-defined amounts to determining a condition on $f,f'$ that guarantees that the two mapping tori are isomorphic as bundles (but not necessarily as bundles with a connection).

\begin{definition}
  Two functions $f,f' \in \Diff^+(\Sigma)$ are \emph{isotopic} if there exists a smooth path $f_t \in \Diff^+(\Sigma)$ such that $f_0 = f, f_1 = f'$.
  Denote by $f \sim f'$ the equivalence relation by isotopy between $f,f' \in \Diff^+(\Sigma)$. 

  The \emph{mapping class group} $\M(\Sigma)$ is the group of isotopy classes in $\Diff^+(\Sigma)$. Equivalently, we have
  \[ \M(\Sigma) := \quotient{\Diff^+(\Sigma)}{\sim} = \pi_0(\Diff^+(\Sigma)) \]

  Two elements $[f],[f'] \in\M(\Sigma)$ are considered \emph{conjugate} if there exists $[g] \in \M(\Sigma)$ such that $[f'] = [g][f][g]^{-1}$. Equivalently, $[f'] = [g\circ f \circ g^{-1}]$.
\end{definition}

\begin{proposition}\label{prop:MappingToriIsoCondition}
  Two mapping tori $\Sigma_f$ and $\Sigma_{f'}$ are bundle isomorphic if and only if $f$ and $f'$ are representatives of conjugate elements $[f],[f']\in \M(\Sigma)$.
\end{proposition}
\begin{proof}
$\implies$: Suppose that there exists a bundle isomorphism $\Psi: \Sigma_f \to \Sigma_{f'}$. For any base point $b \in S^1$ , denote by $\Psi_b$ the induced map from $\Psi$ of the fibre over $b$ in $\Sigma_f$, denoted $\Sigma_{f,b}$, to the fibre over $b$ in $\Sigma_{f'}$, denoted $\Sigma_{f',b}$. Hence $\Psi_b:\Sigma_{f,b} \to \Sigma_{f',b}$ and we can consider $\Psi_b\in \Diff^+(\Sigma)$. We will show that $[f] = [\Psi_0^{-1} \circ f' \circ \Psi_0]$. In order to do this we need to find an isotopy $f_s$ such that $f_0 = f$ and $f_1 = \Psi_0^{-1}\circ f'\circ \Psi_0$. 

Let $\Y_f,\Y_{f'}$ be connections on respectively $\Sigma_f, \Sigma_{f'}$ with monodromy $f, f'$ and denote their respective flows $\Phi^f,\Phi^{f'}$. Consider the map $f_s$ defined as the composition 
\[f_s: \Sigma_{f,0} \xrightarrow{\Psi_0} \Sigma_{f',0} \xrightarrow{\Phi^{f'}(s,\cdot)} \Sigma_{f',s} \xrightarrow{\Psi_s^{-1}}  \Sigma_{f,s} \xrightarrow{\Phi^f(1-s,\cdot)} \Sigma_{f,0}.\]       
Then $f_0 = f$ and $f_1 = \Phi_0^{-1}\circ f' \circ \Phi_0 $ as desired. 

$\impliedby$: We have already established that if $f$ is conjugate to $f'$ the mapping tori are isomorphic. It remains to show if $[f] = [f']$ then $\Sigma_f$ is bundle isomorphic to $\Sigma_{f'}$. Indeed, let $f_s$ be the isotopy such that $f_0 = f$ and $f_1 = f'$. Then the isomorphism is given by 
$$ \Psi: [t,x] \mapsto [t, f_{1-t}\circ f^{-1}(x)]$$
To check it is an isomorphism we need to make sure it is well-defined, that is, $\Psi([2\pi ,x]) = \Psi([0,f(x)])$. We have $\Psi([2\pi,x]) = [2\pi,f_0\circ f^{-1}(x)] = [1,x]$ and $\Psi([0,f(x)]) = [0,f_1\circ f^{-1}\circ f(x)] = [0,f'(x)] = [1,x] = \Psi(1,x)$ as desired. 
\end{proof}

An immediate consequence of \cref{prop:MappingToriIsoCondition} is that the monodromy representation of $M$ is a complete invariant of bundles up to isomorphism.

\begin{theorem}\label{thm:bundlesAreMCG}
  Let $\Sigma$ be a smooth, compact, orientable manifold (possibly with boundary). Then the correspondence, which sends any $\Sigma$-bundle over $S^1$ to its monodromy representation, induces a bijection
  \[ \{ \text{isomorphism classes of $\Sigma$-bundles over $S^1$} \} \cong \{ \text{conjugacy classes of $\M(\Sigma)$} \}. \]
\end{theorem}

\subsection{Properties of the Mapping Class Group of Surfaces}\label{sec:MappingClassGroupOfSurfaces}

Due to \cref{thm:bundlesAreMCG}, two $\Sigma$-bundles over $S^1$ are isomorphic if and only if they have monodromy representations that are conjugate in $\M(\Sigma)$. Consequently, by studying the conjugacy classes of $\M(\Sigma)$ for a particular choice of $\Sigma$, all the possible $\Sigma$-bundles over $S^1$ up to isomorphism can be determined. In particular, if $\M(\Sigma)$ is trivial then all $\Sigma$-bundles are isomorphic to the trivial bundle. 

The mapping class groups for different surfaces are well understood. Thorough accounts are given in \cite{farbPrimerMappingClass2012,ivanovMappingClassGroups2001,birmanBraidsLinksMapping2016}. However, most of the work has been developed for the mapping class group of diffeomorphisms which leave the boundary of the surface \emph{point-wise} invariant. We are interested in diffeomorphisms which only leave the boundary \emph{set-wise} invariant. In this section we review some important elements of the theory of mapping class groups, converting where necessary proofs of well known results in the theory for point-wise invariant boundary to the case of set-wise invariant boundary.

\subsubsection{Curves and Dehn Twists}
A \emph{simple closed curve} in $\Sigma$ is a smooth embedding $c:S^1 \to \Sigma$ which does not intersect the boundary. The curve $c$ has an orientation and the curve with the same image but opposite orientation will be denoted $c^{-1}$.
We will follow \cite{farbPrimerMappingClass2012} by denoting a simple closed curve $c$ as \emph{essential} if it does not bound a disk or it is isotopic to a boundary component.\footnote{Note that this definition differs from that in \cite{parisGeometricSubgroupsMapping2000}.}

We will write $a \simeq b$ if the two simple closed curves $a,b$ are isotopic. 
The \emph{index of intersection} of two simple closed curves $a$ and $b$ is 
\[ I(a,b) = \min\{ |a'\cap b'|\, :\, a'\simeq a, b'\simeq b \}. \]
Through deformations of $a'$ or $b'$ it is always possible to get a finite intersection so that $I(a,b)$ is well-defined. See \cite{fathiThurstonWorkSurfaces2012} for more details.

We will use the following propositions from \cite{parisGeometricSubgroupsMapping2000} about simple closed curves.

\begin{proposition}{\cite[Prop.~3.4]{parisGeometricSubgroupsMapping2000}}
    \label{prop:ExistsMinimalb}
    Suppose that $a_1,\dots, a_m$ are essential simple closed curves such that, if $i\neq j$, then $a_i \cap a_j = \emptyset$ and $a_i$ is neither isotopic to $a_j$ or $a_j^{-1}$. Then, for each $i$, there exists a simple closed curve $b_i$ such that $a_j\cap b_i = \emptyset$ if $i\neq j$, and $|a_i \cap b_i| = I(a_i,b_i) > 0$.
\end{proposition}

\begin{proposition}{\cite[Prop.~3.10]{parisGeometricSubgroupsMapping2000}}
    \label{prop:fixingIsotopyExists}
    Let $a_1\dots, a_m$ and $b_1,\dots, b_m$ be simple closed curves none of which are null-homotopic in $\Sigma$. Suppose that for $i\neq j$:
    \begin{enumerate}
        \item $a_i\cap a_j = \emptyset$ and $b_i \cap b_j = \emptyset$.
        \item $a_i$ is not isotopic to $a_j^{\pm 1}$ and $b_i$ is not isotopic to $b_j^{\pm 1}$.
        \item $a_i$ is isotopic to $b_i$.
    \end{enumerate}
    Then there exists an isotopy of diffeomorphisms $h_t$ on $\Sigma$ such that $h_0 = Id$ and $h_1 \circ a_i = b_i$ for all $i=1,\dots, p$.
\end{proposition}

The Dehn-Likorish Twist Theorem demonstrates that the mapping class group of any closed surface (compact and without boundary) can be finitely presented by so called Dehn twists. Even in the case of surfaces with boundary, Dehn twists play a crucial role of understanding the mapping class group. 

Let $A_{1,2} \subset \mathbb{C}$ be the annulus
$$A_{1,2} = \{z\in\C\, :\, 1\leq |z|\leq 2\}.$$
Define the \emph{standard Dehn twist} on $A_{1,2}$ as the transformation
$$\tau: z \mapsto e^{2\pi i(|z|)}z.$$
Note that $\tau$ is the identity on each boundary component and `twists' the outer boundary one rotation anticlockwise with respect to the inner boundary. Finally, $\tau$ is orientation preserving.
In other words $\tau$ is in the group of orientation preserving diffeomorphisms of $A_{1,2}$ which leave the boundary point-wise invariant.

Let $a$ be a simple closed curve in a surface $\Sigma$. Let $U$ be a regular neighbourhood of $a$. By regular neighbourhood, we mean a closed embedded annulus $U$ of $\Sigma$ such that, if $i:U\to \Sigma$ is the inclusion, then $i^*a$ is a simple closed curve in $U$. Choose an orientation preserving diffeomorphism $\phi: A_{1,2} \to U$. 
We obtain a homeomorphism $\tau_a$ as follows:
$$
\tau_a(x) = \begin{cases}
\phi\circ \tau \circ \phi^{-1}(x) & \text{if $x\in U$} \\
x & \text{if $x\in \Sigma\setminus U$ }.
\end{cases}
$$
In other words, $\tau_a$ performs the standard Dehn twists $\tau$ on the annulus $U$ and fixes every point outside of $U$. Denote the isotopy class of $\tau_a$, namely $T_a:=[\tau_a]$, the \emph{Dehn twist about $a$}.

% Here are some well-established observations about the Dehn twist (see, for instance, \cite{farbPrimerMappingClass2012}).
% \begin{enumerate}
%     \item The Dehn twist along $a^{-1}$ coincides with the Dehn twist along $a$.
%     \item The curve $a$ is fixed by $\tau_a$.
%     \item If two curves are isotopic, then so are their corresponding Dehn twists.
%     \item If $h$ is a diffeomorphism of $\Sigma$, the Dehn twist along $h(a)$ is $[h \tau_a h^{-1}]$.
%     \item If $a$ is not essential then $T_a = 1$, that is, $\tau_a$ is isotopic to the identity.\footnote{This is only true here because we allow the diffeomorphisms to only keep the boundary set-wise invariant. If we look at the usual case (such as in \cite{parisGeometricSubgroupsMapping2000}) where the boundary is kept pointwise invariant, then curves isotopic to a boundary component will not be the identity.}
% \end{enumerate}

We will state some key properties of Dehn twists. The first describes how a Dehn twist changes index of intersection of two simple closed curves.

\begin{proposition}{\cite[Prop.~3.3]{parisGeometricSubgroupsMapping2000}}
    \label{prop:DehnTwistsOnIntersection}
    Let $a,b:S^1 \to \Sigma$ be two simple closed curves and $n$ any integer. Then 
    \[ I(\tau_a^n \circ b,b) = |n| I(a,b)^2. \]
\end{proposition}

The second proposition is a modification of \cite[Prop.~3.8]{parisGeometricSubgroupsMapping2000} to the case of set-wise invariant boundary.

\begin{proposition}{\cite[Prop.~3.8]{parisGeometricSubgroupsMapping2000}}
    \label{prop:TwistsCommute}
    Suppose $a_1,\dots, a_m:S^1 \to \Sigma$ are essential simple closed curves which are pairwise disjoint, and no curve $a_i$ is isotopic to $a_j$ or $a_j^{-1}$, $i\neq j$. Consider the function 
    \[ h:\Z^m \to \M(\Sigma) \]
    defined by 
    \[ h(n_1,\dots, n_m) = T_{a_1}^{n_1}\dots T_{a_m}^{n_m}. \]
    Then $h$ is an injective homomorphism.
\end{proposition}
\begin{proof}
    First, because the curves $a_1,\dots, a_m$ are pairwise disjoint, the Dehn twists commute, and $h$ is a homomorphism.
    \footnote{To see they commute note that if they are disjoint we can find regular neighbourhoods $N_i$ of the respective $a_i$ that are disjoint. Then the Dehn twists for each $a_i$ is the identity out of the respective $N_i$, so the twisting of any two commutes.}
    To see that it is injective, suppose that $T_{a_1}^{n_1} \cdots T_{a_p}^{n_p} = 1$. 
    Fix an index $i$. \cref{prop:ExistsMinimalb} supplies a simple closed curve $b$ in $\Sigma$ disjoint from $a_j,j\neq i$ such that 
    \[ |a_i \cap b| = I(a_i,b) >0.\]
    We calculate, using commutativity of the twists, $b\cap a_j = \emptyset$, $i\neq j$, and  \cref{prop:DehnTwistsOnIntersection}, that 
    \[ 0 = I(b,b) = I(\tau_{a_1}^{n_1}\circ\dots \circ\tau_{a_m}^{n_m}\circ b,b) = I(\tau_{a_i}^{n_i}\circ b,b) = |n_i| I(a_i, b)^2. \]
    Therefore $n_i = 0$ and we have shown $h$ is injective.
\end{proof}

As an immediate consequence of \cref{prop:TwistsCommute} we have the following.

\begin{corollary}{\cite[Cor.~3.9]{parisGeometricSubgroupsMapping2000}}
    \label{cor:InfiniteOrderTwists}
    If $a:S^1 \to \Sigma$ is an essential simple closed curve, then the Dehn twist about $a$ has infinite order in $\M(\Sigma)$.
\end{corollary}

Note that, except for the cases of a disk, annulus, or pair of pants (a disk with two open disks removed), every other compact surface has at least one essential simple closed curve. For these compact surfaces in particular, \cref{cor:InfiniteOrderTwists} implies that $\M(\Sigma)$ contains a subgroup isomorphic to $\Z$. As $\Z$ is abelian, the conjugacy classes are singletons. That is, there are infinitely many conjugacy classes in $\M(\Sigma)$. By \cref{thm:bundlesAreMCG} there is thus infinitely many isomorphism classes of $\Sigma$-bundles over $S^1$ whenever $\Sigma$ is a compact surface that is not a disk, annulus, or pair of pants.

In the case that $\Sigma$ is a disk, annulus, or pair of pants, the mapping class groups are finite. This follows from the following well established result (see \cite{farbPrimerMappingClass2012,birmanBraidsLinksMapping2016,czechowskiSymplectomorphismsDiscreteBraid2017}).

\begin{theorem}\label{thm:MappingClassGroupIsomorphisms}
    The following isomorphisms hold.
    \begin{enumerate}
        \item If $\Sigma$ is a disk $D^2$ then $\M(D^2) \cong 1$.
        \item If $\Sigma$ is an annulus $A_{1,2}$ then $\M(A) \cong 1$.
        \item If $\Sigma$ is a pair of pants $P$ then $\M(P) \cong \Z_2$.
    \end{enumerate}
\end{theorem}

\end{document}